\documentclass[11pt]{article}
\usepackage[utf8]{inputenc}
\usepackage[T1]{fontenc}
\usepackage{amsmath}
\usepackage{amsfonts}
\usepackage{amssymb}
\usepackage{amsthm}
\usepackage{lineno}
\usepackage[version=4]{mhchem}
\usepackage{float}
\usepackage{url}
\usepackage{caption}
\usepackage{natbib}
\usepackage{subcaption}
\usepackage{placeins}
\usepackage{multirow}
\usepackage{stmaryrd}
\usepackage{graphicx} 
\usepackage{booktabs} 
\usepackage[a4paper, margin=1in]{geometry}
\usepackage{algorithm}
\usepackage{algpseudocode}
\usepackage{booktabs}
\usepackage{hyperref}
\theoremstyle{definition}

\newtheorem{theorem}{Theorem}[section]

\newtheorem{proposition}{Proposition}[section]

\newtheorem{corollary}{Corollary}[section]

\graphicspath{ {./figs/} }

\title{Machine Learning Integrated in Wavelet Shrinkage (MLShrink)}
\author{Dixon Vimalajeewa$^{1}$, Vijini Lakmini$^{2}$,  and Brani Vidakovic$^{2}$ \\
\\
{\small $^{1}$ Department of Statistics, University of Nebraska, Lincoln, NE 68588, USA}\\
{\small $^{2}$Department of Statistics, Texas A\&M University, College Station, TX 77843, USA} 
}

\date{\today}

\begin{document}

\maketitle


\begin{abstract}
Data encountered in practice are frequently contaminated by additive noise, and wavelet shrinkage remains a fundamental tool for recovering underlying signals in nonparametric estimation. Classical procedures such as hard and soft thresholding decide whether to retain a wavelet coefficient almost entirely from its magnitude. Although effective in many settings, these rules can be too rigid for coefficients whose magnitudes fall in an intermediate region where the distinction between signal and noise is uncertain.

 We propose \emph{MLShrink}, a two-threshold wavelet denoising procedure that combines wavelet shrinkage with machine learning. Coefficients below a lower threshold are discarded, coefficients above an upper threshold are retained, and coefficients in the intermediate band are classified using local wavelet-domain features. In this way, \emph{MLShrink} preserves the simplicity of classical thresholding away from the decision boundary while allowing data-adaptive decisions for ambiguous coefficients.

The paper also develops a theoretical framework tailored to this architecture. We show that \emph{MLShrink} is a nonexpansive support-selection rule, derive an oracle-based risk decomposition showing that excess denoising risk is determined by classification errors on the undecided band, and establish an oracle-consistency result under suitable assumptions on classifier performance.

Simulation experiments on standard benchmark signals indicate that \emph{MLShrink} is competitive with several established wavelet shrinkage methods and is especially effective for signals with irregular, edge-rich, or non-smooth structure. These findings suggest that learned decisions on the intermediate threshold band provide a useful and interpretable connection between classical wavelet denoising and modern statistical learning.

\end{abstract}
\section{Introduction}

Signal measurements collected in practice are almost always contaminated by additive noise, and a primary task in many scientific and engineering pipelines is therefore to recover a clean estimate of the underlying signal. Over the past several decades, wavelet methods have become a standard tool for this purpose in nonparametric statistics and signal processing because they represent signals locally in both time and frequency. This localization makes wavelet methods especially effective for signals with spatially inhomogeneous behavior, transient phenomena, edges, and abrupt structural changes \citep{vimalajeewa2023}.

A central wavelet-domain denoising strategy is \emph{wavelet shrinkage}, also called \emph{wavelet thresholding}. Following the seminal work of Donoho and Johnstone, wavelet shrinkage has become a cornerstone of nonparametric function estimation and signal denoising \citep{donoho1994ideal,donoho1995adapting,donoho1998minimax}. The basic idea is simple: transform the noisy signal into the wavelet domain, modify the empirical coefficients by a shrinkage rule, and reconstruct the signal by the inverse transform. This approach is effective because many signals admit sparse wavelet representations, so that a relatively small number of coefficients carry most of the structural information, while the noise is spread more diffusely across coefficients.

A broad range of shrinkage rules has been developed, including risk-based procedures, cross-validation methods, multiple-testing formulations, block thresholding rules, and Bayesian approaches \citep{antoniadis2001wavelet}. Important extensions include translation-invariant denoising, adaptive multiple-testing thresholding, firm shrinkage, Bayesian thresholding, block thresholding, neighboring-coefficient rules, and empirical-Bayes threshold selection \citep{CoifmanDonoho1995TI,AbramovichBenjamini1996,GaoBruce1997,AbramovichSapatinasSilverman1998,HallKerkyacharianPicard1998,Cai1999,CaiSilverman2001,JohnstoneSilverman2005}.

In most of these methods, however, the essential decision is still the same: determine which empirical wavelet coefficients should be retained as signal and which should be suppressed as noise. Classical rules such as hard and soft thresholding make this decision using a single threshold. Although these methods are simple and often effective, they can be too rigid near the decision boundary. Very small coefficients are usually dominated by noise, and very large coefficients are usually informative. The main difficulty arises for coefficients whose magnitudes fall in an intermediate region where the distinction between signal and noise is unclear. In that regime, a single-threshold rule may either discard useful structure or retain noisy artifacts.

In this paper we propose \emph{MLShrink}, a two-threshold wavelet denoising procedure that combines classical wavelet shrinkage with machine learning. Instead of relying on a single threshold, \emph{MLShrink} uses two thresholds, $\lambda_1$ and $\lambda_2$, with $0 < \lambda_1 < \lambda_2$. Coefficients with $|d_{j,k}| \le \lambda_1$ are treated as clear noise and discarded, while coefficients with $|d_{j,k}| \ge \lambda_2$ are treated as clear signal and retained. Coefficients in the intermediate band $\lambda_1 < |d_{j,k}| < \lambda_2$ are regarded as \emph{undecided}. Their final status is determined by a classifier trained on the confidently labeled coefficients and informed by local wavelet-domain features.

This formulation is attractive because it separates easy decisions from difficult ones. Coefficients far below the lower threshold and far above the upper threshold require no sophisticated treatment, while only the ambiguous middle band is subjected to additional modeling. In this sense, the novelty of \emph{MLShrink} is localized exactly where classical thresholding is most vulnerable. The method is also naturally flexible: once the two-threshold architecture is fixed, different learning algorithms may be used for the undecided band without changing the overall shrinkage mechanism.

\emph{MLShrink} is related to earlier multithreshold ideas, including semi-soft thresholding and previously proposed semi-supervised wavelet shrinkage methods. The present method differs in how the intermediate-band coefficients are handled. Semi-soft thresholding replaces the abrupt threshold decision by a deterministic linear transition, whereas earlier semi-supervised wavelet shrinkage methods rely on a specific manifold-regularization mechanism. In contrast, \emph{MLShrink} treats the middle-band decision as a classification problem and allows a broader family of learning algorithms, including logistic regression, support vector machines, decision trees, random forests, and neural networks. 
On the methodological side, \emph{MLShrink} is especially close in spirit to firm shrinkage, neighboring-coefficient methods, and Bayesian wavelet thresholding, but differs in replacing a deterministic or model-specific middle-band rule by a classifier-based decision on the undecided band \citep{GaoBruce1997,CaiSilverman2001,AbramovichSapatinasSilverman1998}.
This preserves the underlying wavelet shrinkage architecture while allowing more adaptive decisions in the most ambiguous region. The same viewpoint also leads naturally to the theory developed later in the paper: outside the intermediate band, \emph{MLShrink} behaves like an ordinary deterministic thresholding rule, while its statistical gain or loss is governed by the quality of classification on the undecided band.

To assess empirical performance, we study \emph{MLShrink} on several standard benchmark signals observed under multiple noise levels and compare it with a collection of established wavelet shrinkage methods. The numerical results indicate that \emph{MLShrink} is competitive across a range of settings and is especially promising for signals with irregular, edge-rich, or non-smooth structure. These findings suggest that learned decisions on the intermediate threshold band provide a useful and interpretable bridge between classical wavelet denoising and modern statistical learning.

The remainder of the paper is organized as follows. Section~\ref{sec_prelm_methods} reviews the wavelet-domain background needed for the development of \emph{MLShrink}. Section~\ref{sec_MLShrink} introduces the \emph{MLShrink} procedure. Section~\ref{sec_theory}  develops its theoretical properties. Section~\ref{sec_perfom_eval}  presents parameter selection, simulation design, and comparative performance evaluation. Section~\ref{sec_discussion} discusses the implications of the results and possible extensions. Section~\ref{sec_conclusion} concludes the paper.


\section{Preliminaries}\label{sec_prelm_methods}
This section briefly reviews the statistical model, the wavelet-domain representation, and the main thresholding ideas that motivate \emph{MLShrink}.
Suppose we observe noisy data
\begin{eqnarray}
y_i = f(t_i) + \sigma \varepsilon_i, \qquad i=1,2,\ldots,n,
\label{eq:model_time}
\end{eqnarray}

where $f$ is the unknown regression function of interest, $\sigma>0$ is the noise level, and the errors $\varepsilon_i$ are independent standard normal random variables. The goal is to estimate $f$ from the contaminated observations with small mean squared error. For an estimator $\hat f$, we measure performance by the empirical $L_2$ risk
\begin{eqnarray}
R(f,\hat f) = \frac{1}{n}\sum_{i=1}^n E\Bigl(f(t_i)-\hat f(t_i)\Bigr)^2.
\label{eq:risk_prelim}
\end{eqnarray}

Let $W$ denote an orthonormal discrete wavelet transform matrix, and define the empirical wavelet coefficients and true wavelet coefficients by
\begin{eqnarray}
d = Wy, \qquad \theta = Wf.
\end{eqnarray}

Because $W$ is orthonormal, the Gaussian noise structure is preserved in the transform domain, so the model becomes
\begin{eqnarray}
d_{j,k} = \theta_{j,k} + \sigma z_{j,k},
\label{eq:gaussian_sequence}
\end{eqnarray}

where the $z_{j,k}$ are again standard normal random variables. This representation is fundamental in wavelet denoising: many signals have sparse wavelet expansions, whereas noise is spread more diffusely across coefficients. As a result, denoising can be carried out by shrinking or thresholding the empirical coefficients $d_{j,k}$ and then reconstructing the signal by the inverse transform.
The simplest and most widely used rule is hard thresholding. Given a threshold $\lambda>0$, the hard-threshold estimator is
\begin{eqnarray}
\hat d^{\,H}_{j,k} =
\left\{
\begin{array}{ll}
d_{j,k}, & |d_{j,k}| \ge \lambda, \\[1ex]
0, & |d_{j,k}| < \lambda.
\end{array}
\right.
\label{eq:hard_prelim}
\end{eqnarray}

A common default choice is the universal threshold
\begin{eqnarray}
\lambda = \hat\sigma \sqrt{2\log n},
\label{eq:universal_prelim}
\end{eqnarray}

where $\hat\sigma$ is a noise estimate, often obtained from the finest-scale detail 
\citep{donoho1994ideal}. The corresponding signal estimate is then reconstructed as
\begin{eqnarray}
\hat y = W^T \hat d.
\end{eqnarray}

Hard thresholding is attractive because of its simplicity and its ability to preserve large coefficients exactly. Its main weakness is that the keep-or-kill decision is abrupt near the threshold, which can lead to instability and loss of moderate but informative coefficients.

To reduce this rigidity, several multithreshold rules have been proposed. One important example is semi-soft thresholding, also known as firm shrinkage. Given two thresholds $0<\lambda_1<\lambda_2$, the semi-soft rule is
\begin{eqnarray}
S(d_{j,k})=
\left\{
\begin{array}{ll}
0, & |d_{j,k}| \le \lambda_1, \\[1.2ex]
\operatorname{sign}(d_{j,k})
\frac{\lambda_2\bigl(|d_{j,k}|-\lambda_1\bigr)}{\lambda_2-\lambda_1},
& \lambda_1 < |d_{j,k}| < \lambda_2, \\[2ex]
d_{j,k}, & |d_{j,k}| \ge \lambda_2.
\end{array}
\right.
\label{eq:semisoft_prelim}
\end{eqnarray}

This rule interpolates between hard and soft thresholding. Coefficients below $\lambda_1$ are discarded, coefficients above $\lambda_2$ are retained, and coefficients in the transition region are modified by a deterministic linear rule. The advantage is continuity; the limitation is that the intermediate-band decision still depends only on magnitude. To reduce this rigidity, several multithreshold and nonconvex shrinkage rules have been proposed. One important example is semi-soft thresholding, also known as firm shrinkage. Another influential example is SCAD-type wavelet thresholding, which also introduces a more flexible transition between suppression and retention than classical hard or soft thresholding \citep{GaoBruce1997,AntoniadisFan2001,KudryavtsevShestakov2024,Kulkarnietal2026}.

A second line of work, closer in spirit to the present paper, treats the intermediate region as a classification problem. More broadly, the idea of borrowing strength from nearby coefficients has substantial precedent in neighboring-coefficient and block-thresholding methods, where local context is used to stabilize coefficient selection and improve risk performance \citep{HallKerkyacharianPicard1998,Cai1999,CaiSilverman2001}. Coefficients far below a lower threshold are labeled as noise, coefficients far above an upper threshold are labeled as signal, and coefficients in the middle band are treated as uncertain. Earlier semi-supervised wavelet shrinkage methods resolve these uncertain coefficients using manifold-regularization ideas and neighborhood information. This viewpoint is important because it suggests that the main statistical difficulty lies not in the clearly small or clearly large coefficients, but in the ambiguous middle band. \emph{MLShrink} builds directly on this perspective by replacing a fixed deterministic rule on the undecided band with a learned decision rule based on local wavelet-domain features.

\section{MLShrink}\label{sec_MLShrink}

We now introduce \emph{MLShrink}, a two-threshold wavelet shrinkage procedure in which the coefficients that are clearly small or clearly large are handled deterministically, while coefficients in an intermediate band are classified by a learning algorithm.

Let $d_{j,k}$ denote a detail wavelet coefficient at resolution level $j$ and location $k$. The \emph{MLShrink} rule is based on two thresholds,
\begin{eqnarray}
\lambda_1 = \hat{\sigma}\sqrt{c\log n}, \qquad
\lambda_2 = \hat{\sigma}\sqrt{2\log n}, \qquad 0<c<2,
\label{eq:MLShrink_thresholds}
\end{eqnarray}

where $\hat{\sigma}$ is an estimate of the noise standard deviation. The restriction $0<c<2$ guarantees that $\lambda_1<\lambda_2$, so that a genuine undecided band is present.

Thus, every empirical detail coefficient is assigned to one of three regions (see Figure~\ref{fig:MLShrink}):
\begin{eqnarray}
\mathcal{R}_0 &=& \{(j,k): |d_{j,k}| \le \lambda_1\}, \nonumber\\
\mathcal{R}_u &=& \{(j,k): \lambda_1 < |d_{j,k}| < \lambda_2\}, \label{eq:MLShrink_regions}\\
\mathcal{R}_1 &=& \{(j,k): |d_{j,k}| \ge \lambda_2\}. \nonumber
\end{eqnarray}

The set $\mathcal{R}_0$ contains coefficients treated as noise, $\mathcal{R}_1$ contains coefficients treated as signal, and $\mathcal{R}_u$ is the undecided band.
\begin{figure}
	\centering
\includegraphics[width=.8\columnwidth]{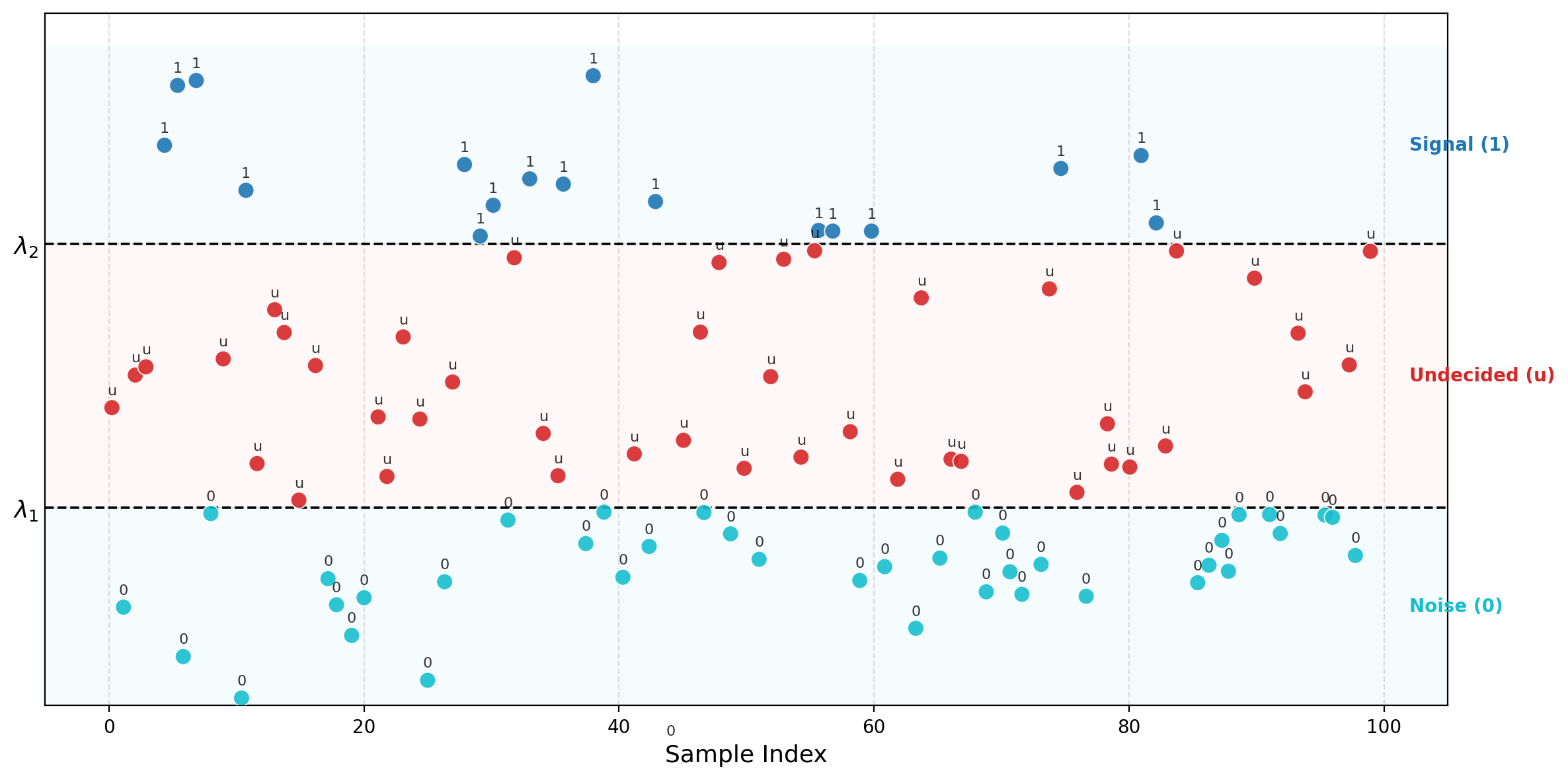}
	\caption{Schematic representation of the \emph{MLShrink} rule. The thresholds $\lambda_1$ and $\lambda_2$ divide empirical wavelet coefficients into three groups: coefficients below $\lambda_1$ are treated as noise, coefficients above $\lambda_2$ are treated as signal, and coefficients in the intermediate band are classified by a learning algorithm.}
\label{fig:MLShrink}
\end{figure}

It is convenient to define the initial label
$$
L_{j,k}^{(0)} =
\left\{
\begin{array}{ll}
0, & (j,k)\in\mathcal{R}_0, \\[1ex]
u, & (j,k)\in\mathcal{R}_u, \\[1ex]
1, & (j,k)\in\mathcal{R}_1.
\end{array}
\right.
$$

Equivalently,
$$
L_{j,k}^{(0)} =
\left\{
\begin{array}{ll}
0, & |d_{j,k}| \le \lambda_1, \\[1ex]
u, & \lambda_1 < |d_{j,k}| < \lambda_2, \\[1ex]
1, & |d_{j,k}| \ge \lambda_2.
\end{array}
\right.
$$

The clearly labeled coefficients form the training set,
\begin{eqnarray}
\mathcal{D}_{\rm lab}
=
\{(d_{j,k},L_{j,k}^{(0)}): L_{j,k}^{(0)}\in\{0,1\}\},
\end{eqnarray}

and the ambiguous coefficients form the undecided set,
\begin{eqnarray}
\mathcal{D}_{u}
=
\{(d_{j,k},L_{j,k}^{(0)}): L_{j,k}^{(0)}=u\}.
\end{eqnarray}

For each coefficient $d_{j,k}$ we construct a feature vector
\begin{eqnarray}
x_{j,k} = \Bigl(|d_{j,k}|,\, j,\, \nu_{j,k}\Bigr)^T \in \mathbb{R}^3,
\label{eq:feature_vector}
\end{eqnarray}

where the local neighborhood summary is defined by
\begin{eqnarray}
\nu_{j,k} = \frac{|d_{j,k-1}| + |d_{j,k+1}|}{2}.
\label{eq:local_summary}
\end{eqnarray}

Thus the features incorporate the coefficient magnitude, the resolution level, and a same-scale local context term. For boundary coefficients, one may use the available neighbor only or a simple boundary convention; this has negligible impact on the implementation.

Let

\begin{eqnarray}
X = \{x_{j,k}\}, \qquad
X_{\rm train} = \{x_{j,k}:(d_{j,k},L_{j,k}^{(0)})\in\mathcal{D}_{\rm lab}\},
\qquad
Y_{\rm train} = \{L_{j,k}^{(0)}:(d_{j,k},L_{j,k}^{(0)})\in\mathcal{D}_{\rm lab}\}.
\end{eqnarray}

A classifier $\mathcal{C}$ is trained on the confidently labeled coefficients:
\begin{eqnarray}
\mathcal{C}: \mathbb{R}^3 \longrightarrow \{0,1\}.
\end{eqnarray}

Although the labels are generated by the two-threshold rule, the learning step itself is an ordinary supervised classification problem carried out on the labeled subset. For each undecided coefficient $(j,k)\in R_u$, the trained classifier produces the predicted label
$$
\hat L_{j,k} = C(x_{j,k}).
$$

For notational convenience, define the final label for all coefficients by
$$
\tilde L_{j,k} =
\left\{
\begin{array}{ll}
0, & (j,k)\in R_0,\\
\hat L_{j,k}, & (j,k)\in R_u,\\
1, & (j,k)\in R_1.
\end{array}
\right.
$$

The \emph{MLShrink} estimator in the wavelet domain is then
\begin{eqnarray}
\hat d^{\,MLShrink}_{j,k} = d_{j,k}\tilde L_{j,k}.
\end{eqnarray}

Equivalently,
\begin{eqnarray}
\hat d^{\,MLShrink}_{j,k}
=
d_{j,k}{\bf 1}\{|d_{j,k}|\ge \lambda_2\}
+
d_{j,k}{\bf 1}\{\lambda_1<|d_{j,k}|<\lambda_2\}\hat L_{j,k}.
\end{eqnarray}

Hence \emph{MLShrink} coincides with an ordinary two-threshold keep-discard rule outside the intermediate band and differs from classical thresholding only through the learned decisions on $R_u$. After thresholding the detail coefficients, the denoised signal is reconstructed by the inverse wavelet transform,
\begin{eqnarray}
\hat y = W^T \hat d^{\,MLShrink},
\end{eqnarray}

while the scaling coefficients are left unchanged. The procedure is classifier-agnostic: in our experiments we consider Logistic Regression (\emph{LR}), Support Vector Machines (\emph{SVM}), Decision Trees (\emph{DT}), Random Forests (\emph{RF}), and Neural Networks (\emph{NN}), but the overall shrinkage architecture remains unchanged. For completeness and reproducibility, a step-by-step
algorithmic description of the \emph{MLShrink} procedure is
provided in Algorithm~S1 in the Supplementary Information.

The procedure is classifier-agnostic. In our experiments we consider logistic regression, support vector machines, decision trees, random forests, and neural networks. The specific choice of classifier affects only the prediction rule on the undecided band; the overall shrinkage architecture remains unchanged. This formulation is particularly convenient for the theoretical development in the next section, since the contribution of \emph{MLShrink} can be isolated to the classification of coefficients with $\lambda_1<|d_{j,k}|<\lambda_2$.

\section{Theoretical Properties of MLShrink}\label{sec_theory}

In this section we give a theoretical characterization of \emph{MLShrink}. The main point is that \emph{MLShrink} is not a new continuous shrinkage family, but rather a two-threshold support-selection rule in which the only genuinely nontrivial statistical decision is made on the undecided band. Accordingly, the most natural theory is one that isolates the role of classification on that band and quantifies how classification quality affects denoising risk.

We work in the Gaussian sequence model associated with an orthonormal wavelet transform:
\begin{eqnarray}
d_{j,k} = \theta_{j,k} + \sigma z_{j,k},
\qquad
z_{j,k}\sim N(0,1),
\label{eq:theory_model}
\end{eqnarray}

where the $z_{j,k}$ are independent, or approximately independent under the orthonormal transform. Let $\lambda_1$ and $\lambda_2$ satisfy
\begin{eqnarray}
0 < \lambda_1 < \lambda_2.
\label{eq:theory_threshold_order}
\end{eqnarray}

As in Section~\ref{sec_MLShrink}, define the three coefficient regions
\begin{eqnarray}
R_0 &=& \{(j,k): |d_{j,k}| \le \lambda_1\}, \nonumber\\
R_u &=& \{(j,k): \lambda_1 < |d_{j,k}| < \lambda_2\}, \label{eq:theory_regions}\\
R_1 &=& \{(j,k): |d_{j,k}| \ge \lambda_2\}. \nonumber
\end{eqnarray}

Then the \emph{MLShrink} estimator can be written as
\begin{eqnarray}
\hat d_{j,k}^{\,MLShrink}
=
d_{j,k}{\bf 1}\{|d_{j,k}| \ge \lambda_2\}
+
d_{j,k}{\bf 1}\{\lambda_1 < |d_{j,k}| < \lambda_2\}\hat L_{j,k},
\label{eq:mlshrink_theory_form}
\end{eqnarray}

where $\hat L_{j,k}\in\{0,1\}$ is the predicted label for coefficients in the undecided band.

\subsection{Basic structural properties}

\begin{proposition}\label{prop:mlshrink_structure}
For every coefficient $(j,k)$, the MLShrink estimator in \eqref{eq:mlshrink_theory_form} satisfies
\begin{eqnarray}
\bigl|\hat d_{j,k}^{\,MLShrink}\bigr| \le |d_{j,k}|.
\label{eq:nonexpansive}
\end{eqnarray}

Hence MLShrink is nonexpansive in magnitude. Moreover, if the predicted label $\hat L_{j,k}$ depends only on sign-invariant features such as $|d_{j,k}|$, $j$, and local magnitude summaries, then MLShrink is sign-preserving in the sense that
\begin{equation}
\begin{split}
\hat d_{j,k}^{\,MLShrink}=0
\qquad\mbox{or}\qquad
\operatorname{sgn}\Bigl(\hat d_{j,k}^{\,MLShrink}\Bigr)=\operatorname{sgn}(d_{j,k}).
\label{eq:sign_preserving}
 \end{split}
    \end{equation}
    
Finally, \emph{MLShrink} coincides exactly with the deterministic two-threshold keep-discard rule outside the undecided band $R_u$.
\end{proposition}

\begin{proof}
If $(j,k)\in R_0$, then $\hat d_{j,k}^{\,MLShrink}=0$, so \eqref{eq:nonexpansive} is immediate. If $(j,k)\in R_1$, then $\hat d_{j,k}^{\,MLShrink}=d_{j,k}$, so equality holds in \eqref{eq:nonexpansive}. If $(j,k)\in R_u$, then $\hat d_{j,k}^{\,MLShrink}=d_{j,k}\hat L_{j,k}$ with $\hat L_{j,k}\in\{0,1\}$, hence
$$
\bigl|\hat d_{j,k}^{\,MLShrink}\bigr|
=
|d_{j,k}|\,|\hat L_{j,k}|
\le |d_{j,k}|.
$$

This proves nonexpansiveness. The sign-preserving property follows because multiplication by $\hat L_{j,k}\in\{0,1\}$ either annihilates the coefficient or leaves its sign unchanged. The final statement follows directly from the definition of the estimator.
\end{proof}

Proposition \ref{prop:mlshrink_structure} clarifies that \emph{MLShrink} is a support-selection rule rather than a coefficient-rescaling rule. The only place where it differs from ordinary two-threshold thresholding is the middle band $R_u$.

\subsection{Oracle benchmark and exact risk decomposition}

To measure the quality of classification on the undecided band, it is convenient to introduce an oracle benchmark. For $(j,k)\in R_u$, define the oracle label
\begin{eqnarray}
L^\star_{j,k}
=
\arg\min_{\ell\in\{0,1\}}
\Bigl(\ell d_{j,k}-\theta_{j,k}\Bigr)^2.
\label{eq:oracle_label}
\end{eqnarray}

Equivalently,
\begin{eqnarray}
L^\star_{j,k}
=
\left\{
\begin{array}{ll}
1, & (d_{j,k}-\theta_{j,k})^2 \le \theta_{j,k}^2, \\[1ex]
0, & (d_{j,k}-\theta_{j,k})^2 > \theta_{j,k}^2.
\end{array}
\right.
\label{eq:oracle_label_explicit}
\end{eqnarray}

Thus the oracle retains a middle-band coefficient if keeping it yields smaller local squared error than discarding it.

The associated oracle estimator is
\begin{eqnarray}
\hat d_{j,k}^{\,oracle}
=
d_{j,k}{\bf 1}\{|d_{j,k}| \ge \lambda_2\}
+
d_{j,k}{\bf 1}\{\lambda_1 < |d_{j,k}| < \lambda_2\}L^\star_{j,k}.
\label{eq:oracle_estimator}
\end{eqnarray}

For $(j,k)\in R_u$, define the local oracle loss gap
\begin{eqnarray}
\Delta_{j,k}
=
\Bigl|
\theta_{j,k}^2-(d_{j,k}-\theta_{j,k})^2
\Bigr|
=
\Bigl|
2d_{j,k}\theta_{j,k}-d_{j,k}^2
\Bigr|.
\label{eq:local_gap}
\end{eqnarray}

\begin{theorem}[Exact oracle-risk decomposition]\label{thm:oracle_decomposition}
The total quadratic risk of \emph{MLShrink} admits the exact decomposition
\begin{eqnarray}
E\Bigl\|
\hat d^{\,MLShrink}-\theta
\Bigr\|^2
=
E\Bigl\|
\hat d^{\,oracle}-\theta
\Bigr\|^2
+
\sum_{j,k}
E\Bigl[
{\bf 1}\{(j,k)\in R_u\}
\Delta_{j,k}
{\bf 1}\{\hat L_{j,k}\neq L^\star_{j,k}\}
\Bigr].
\label{eq:risk_decomposition_exact}
\end{eqnarray}

In particular, the excess risk of \emph{MLShrink} relative to the oracle rule is driven entirely by classification errors on the undecided band.
\end{theorem}

\begin{proof}
Outside $R_u$, the \emph{MLShrink} estimator and the oracle estimator coincide coefficientwise. Hence any risk difference can only arise from coefficients in $R_u$.

For a fixed $(j,k)\in R_u$, the \emph{MLShrink} loss is
$$
\Bigl(\hat L_{j,k}d_{j,k}-\theta_{j,k}\Bigr)^2,
$$
while the oracle loss is
$$
\Bigl(L^\star_{j,k}d_{j,k}-\theta_{j,k}\Bigr)^2.
$$

If $\hat L_{j,k}=L^\star_{j,k}$, the difference is zero. If $\hat L_{j,k}\neq L^\star_{j,k}$, then one label yields local loss $\theta_{j,k}^2$ and the other yields local loss $(d_{j,k}-\theta_{j,k})^2$. Since $L^\star_{j,k}$ is defined as the minimizing label, the excess local loss is exactly
$$
\Bigl|
\theta_{j,k}^2-(d_{j,k}-\theta_{j,k})^2
\Bigr|
=
\Delta_{j,k}.
$$

Therefore, coefficientwise,
 $$
\Bigl(\hat d_{j,k}^{\,MLShrink}-\theta_{j,k}\Bigr)^2
=
\Bigl(\hat d_{j,k}^{\,oracle}-\theta_{j,k}\Bigr)^2
+
{\bf 1}\{(j,k)\in R_u\}
\Delta_{j,k}
{\bf 1}\{\hat L_{j,k}\neq L^\star_{j,k}\}.
$$ 

Summing over $(j,k)$ and taking expectations gives \eqref{eq:risk_decomposition_exact}.
\end{proof}

\begin{corollary}[A probability bound for excess oracle risk]\label{cor:oracle_bound}
Suppose that on the undecided band one has
\begin{eqnarray}
|\theta_{j,k}| \le M
\qquad
\mbox{for all } (j,k)\in R_u.
\label{eq:theta_bounded}
\end{eqnarray}

Then
\begin{eqnarray}
E\Bigl\|
\hat d^{\,MLShrink}-\theta
\Bigr\|^2
-
E\Bigl\|
\hat d^{\,oracle}-\theta
\Bigr\|^2
\le
\lambda_2(2M+\lambda_2)
\sum_{j,k}
P\Bigl(
\hat L_{j,k}\neq L^\star_{j,k},\,
(j,k)\in R_u
\Bigr).
\label{eq:risk_bound_probability}
\end{eqnarray}
\end{corollary}

\begin{proof}
If $(j,k)\in R_u$, then $|d_{j,k}|<\lambda_2$. Hence, by \eqref{eq:local_gap},
\begin{eqnarray}
\Delta_{j,k}
&=&
|2d_{j,k}\theta_{j,k}-d_{j,k}^2|
\nonumber\\
&\le&
2|d_{j,k}||\theta_{j,k}|+|d_{j,k}|^2
\nonumber\\
&\le&
2\lambda_2 M+\lambda_2^2
\nonumber\\
&=&
\lambda_2(2M+\lambda_2).
\label{eq:gap_bound}
\end{eqnarray}

Substituting this bound into \eqref{eq:risk_decomposition_exact} yields \eqref{eq:risk_bound_probability}.
\end{proof}

\subsection{Comparison with hard thresholding}

The oracle-risk decomposition compares \emph{MLShrink} with the best possible keep-or-discard rule on the undecided band. It is also useful to compare \emph{MLShrink} directly with a standard deterministic benchmark, namely hard thresholding at level $\lambda_2$.

\begin{theorem}[Exact comparison with hard thresholding on the undecided band]
\label{thm:hard_exact}
Define the hard-threshold estimator at level $\lambda_2$ by
\begin{eqnarray}
\hat d_{j,k}^{\,hard}
=
d_{j,k}{\bf 1}\{|d_{j,k}|\ge \lambda_2\}.
\label{eq:hard_estimator}
\end{eqnarray}

Then
\begin{eqnarray}
&&
E\Bigl\|
\hat d^{\,MLShrink}-\theta
\Bigr\|^2
-
E\Bigl\|
\hat d^{\,hard}-\theta
\Bigr\|^2
\nonumber\\
&=&
-\sum_{j,k}
E\Bigl[
{\bf 1}\{(j,k)\in R_u\}
\Delta_{j,k}
{\bf 1}\{L^\star_{j,k}=1\}
\hat L_{j,k}
\Bigr]
\nonumber\\
&&
+
\sum_{j,k}
E\Bigl[
{\bf 1}\{(j,k)\in R_u\}
\Delta_{j,k}
{\bf 1}\{L^\star_{j,k}=0\}
\hat L_{j,k}
\Bigr].
\label{eq:hard_exact_decomp}
\end{eqnarray}

In particular, \emph{MLShrink} improves on hard thresholding precisely through coefficients in the undecided band for which retention is preferable and the classifier predicts $\hat L_{j,k}=1$.
\end{theorem}

\begin{proof}
Outside $R_u$, the estimators $\hat d^{\,MLShrink}$ and $\hat d^{\,hard}$ coincide coefficientwise. Indeed, both estimators set coefficients in $R_0$ to zero and both retain coefficients in $R_1$. Hence any risk difference arises only from coefficients in $R_u$.

Fix $(j,k)\in R_u$. On this band,
\begin{equation}
    \begin{split}
   \hat d_{j,k}^{\,hard}=0,
\qquad
\hat d_{j,k}^{\,MLShrink}=\hat L_{j,k}d_{j,k},
\qquad
\hat L_{j,k}\in\{0,1\}.
\label{eq:hard_local_forms}
       \end{split}
\end{equation}

Therefore the local risk difference is
\begin{eqnarray}
\Bigl(\hat d_{j,k}^{\,MLShrink}-\theta_{j,k}\Bigr)^2
-
\Bigl(\hat d_{j,k}^{\,hard}-\theta_{j,k}\Bigr)^2
&=&
\Bigl(\hat L_{j,k}d_{j,k}-\theta_{j,k}\Bigr)^2-\theta_{j,k}^2.
\label{eq:local_hard_diff_start}
\end{eqnarray}

Since $\hat L_{j,k}\in\{0,1\}$, this becomes
\begin{eqnarray}
\Bigl(\hat L_{j,k}d_{j,k}-\theta_{j,k}\Bigr)^2-\theta_{j,k}^2
&=&
\hat L_{j,k}\Bigl((d_{j,k}-\theta_{j,k})^2-\theta_{j,k}^2\Bigr).
\label{eq:local_hard_diff}
\end{eqnarray}

If $L^\star_{j,k}=1$, then $(d_{j,k}-\theta_{j,k})^2\le \theta_{j,k}^2$, so
\begin{eqnarray}
(d_{j,k}-\theta_{j,k})^2-\theta_{j,k}^2
=
-\Delta_{j,k}.
\label{eq:local_hard_keep}
\end{eqnarray}

If $L^\star_{j,k}=0$, then $(d_{j,k}-\theta_{j,k})^2>\theta_{j,k}^2$, so
\begin{eqnarray}
(d_{j,k}-\theta_{j,k})^2-\theta_{j,k}^2
=
+\Delta_{j,k}.
\label{eq:local_hard_discard}
\end{eqnarray}

Substituting \eqref{eq:local_hard_keep} and \eqref{eq:local_hard_discard} into \eqref{eq:local_hard_diff} yields
\begin{eqnarray}
&&
\Bigl(\hat d_{j,k}^{\,MLShrink}-\theta_{j,k}\Bigr)^2
-
\Bigl(\hat d_{j,k}^{\,hard}-\theta_{j,k}\Bigr)^2
\nonumber\\
&=&
-\Delta_{j,k}{\bf 1}\{L^\star_{j,k}=1\}\hat L_{j,k}
+
\Delta_{j,k}{\bf 1}\{L^\star_{j,k}=0\}\hat L_{j,k}.
\end{eqnarray}

Multiplying by ${\bf 1}\{(j,k)\in R_u\}$, summing over $(j,k)$, and taking expectations gives \eqref{eq:hard_exact_decomp}.
\end{proof}

\begin{corollary}[A sufficient condition for improvement over hard thresholding]
\label{cor:hard_improvement}

Let
\begin{eqnarray}
S_u = \{(j,k)\in R_u : L^\star_{j,k}=1\}
\label{eq:Su_definition}
\end{eqnarray}

denote the recoverable part of the undecided band. Assume that there exist deterministic constants $\underline{\Delta}_{j,k}>0$ for $(j,k)\in S_u$ and $\overline{\Delta}_{j,k}>0$ for $(j,k)\in R_u\setminus S_u$ such that
\begin{eqnarray}
E\Bigl(\Delta_{j,k}\mid \hat L_{j,k}=1\Bigr) \ge \underline{\Delta}_{j,k},
\qquad (j,k)\in S_u,
\label{eq:delta_lower}
\end{eqnarray}

and
\begin{eqnarray}
E\Bigl(\Delta_{j,k}\mid \hat L_{j,k}=1\Bigr) \le \overline{\Delta}_{j,k},
\qquad (j,k)\in R_u\setminus S_u.
\label{eq:delta_upper}
\end{eqnarray}

Assume also that
\begin{eqnarray}
P(\hat L_{j,k}=1)\ge 1-\eta_1,
\qquad (j,k)\in S_u,
\label{eq:true_keep_prob}
\end{eqnarray}

and
\begin{eqnarray}
P(\hat L_{j,k}=1)\le \eta_0,
\qquad (j,k)\in R_u\setminus S_u.
\label{eq:false_keep_prob}
\end{eqnarray}

Then
\begin{eqnarray}
&&
E\Bigl\|
\hat d^{\,MLShrink}-\theta
\Bigr\|^2
-
E\Bigl\|
\hat d^{\,hard}-\theta
\Bigr\|^2
\nonumber\\
&\le&
-\sum_{(j,k)\in S_u}\underline{\Delta}_{j,k}(1-\eta_1)
+
\sum_{(j,k)\in R_u\setminus S_u}\overline{\Delta}_{j,k}\eta_0.
\label{eq:hard_improvement_bound}
\end{eqnarray}

In particular, if
\begin{eqnarray}
\sum_{(j,k)\in S_u}\underline{\Delta}_{j,k}(1-\eta_1)
>
\sum_{(j,k)\in R_u\setminus S_u}\overline{\Delta}_{j,k}\eta_0,
\label{eq:hard_improvement_condition}
\end{eqnarray}

then
\begin{eqnarray}
E\Bigl\|
\hat d^{\,MLShrink}-\theta
\Bigr\|^2
<
E\Bigl\|
\hat d^{\,hard}-\theta
\Bigr\|^2.
\label{eq:hard_improvement_conclusion}
\end{eqnarray}
\end{corollary}

\begin{proof}
By Theorem~\ref{thm:hard_exact},
\begin{eqnarray}
&&
E\Bigl\|
\hat d^{\,MLShrink}-\theta
\Bigr\|^2
-
E\Bigl\|
\hat d^{\,hard}-\theta
\Bigr\|^2
\nonumber\\
&=&
-\sum_{(j,k)\in S_u}
E\Bigl[\Delta_{j,k}\hat L_{j,k}\Bigr]
+
\sum_{(j,k)\in R_u\setminus S_u}
E\Bigl[\Delta_{j,k}\hat L_{j,k}\Bigr].
\label{eq:hard_cor_start}
\end{eqnarray}

For $(j,k)\in S_u$,
\begin{eqnarray}
E\Bigl[\Delta_{j,k}\hat L_{j,k}\Bigr]
&=&
E\Bigl[\Delta_{j,k}\mid \hat L_{j,k}=1\Bigr]P(\hat L_{j,k}=1)
\nonumber\\
&\ge&
\underline{\Delta}_{j,k}(1-\eta_1),
\end{eqnarray}

by \eqref{eq:delta_lower} and \eqref{eq:true_keep_prob}. Similarly, for $(j,k)\in R_u\setminus S_u$,
\begin{eqnarray}
E\Bigl[\Delta_{j,k}\hat L_{j,k}\Bigr]
&=&
E\Bigl[\Delta_{j,k}\mid \hat L_{j,k}=1\Bigr]P(\hat L_{j,k}=1)
\nonumber\\
&\le&
\overline{\Delta}_{j,k}\eta_0,
\end{eqnarray}
by \eqref{eq:delta_upper} and \eqref{eq:false_keep_prob}. Substituting these bounds into \eqref{eq:hard_cor_start} proves \eqref{eq:hard_improvement_bound}, and \eqref{eq:hard_improvement_conclusion} follows immediately from \eqref{eq:hard_improvement_condition}.
\end{proof}
\subsection{Oracle consistency}

Theorem~\ref{thm:oracle_decomposition} suggests the appropriate asymptotic target. Since the only novel part of \emph{MLShrink} is the classifier on $R_u$, the natural notion of asymptotic optimality is oracle consistency relative to the same two-threshold architecture.

\begin{theorem}[Oracle consistency of \emph{MLShrink}]
\label{thm:oracle_consistency}
Consider a sequence of problems indexed by $n$. Let
\begin{eqnarray}
p_n
=
\sup_{j,k}
P\Bigl(
\hat L_{j,k}^{(n)} \neq L_{j,k}^{\star,(n)}
\,\Big|\,
(j,k)\in R_u^{(n)}
\Bigr).
\label{eq:pn_def}
\end{eqnarray}

Assume that there exist deterministic sequences $M_n$ and $a_n>0$ such that
\begin{eqnarray}
|\theta_{j,k}^{(n)}| \le M_n
\qquad
\mbox{for all } (j,k)\in R_u^{(n)},
\label{eq:Mn_bound}
\end{eqnarray}
\begin{eqnarray}
\lambda_{2,n}(2M_n+\lambda_{2,n})
\sum_{j,k}
P\Bigl(
(j,k)\in R_u^{(n)}
\Bigr)
=
O(a_n),
\label{eq:an_condition}
\end{eqnarray}

and
\begin{eqnarray}
p_n \longrightarrow 0
\qquad\mbox{as}\qquad
n\to\infty.
\label{eq:classifier_consistency}
\end{eqnarray}

Then
\begin{eqnarray}
\frac{
E\Bigl\|
\hat d_n^{\,MLShrink}-\theta_n
\Bigr\|^2
-
E\Bigl\|
\hat d_n^{\,oracle}-\theta_n
\Bigr\|^2
}{a_n}
\longrightarrow 0.
\label{eq:oracle_consistency}
\end{eqnarray}
\end{theorem}

\begin{proof}
By Corollary~\ref{cor:oracle_bound},
\begin{eqnarray}
&&
E\Bigl\|
\hat d_n^{\,MLShrink}-\theta_n
\Bigr\|^2
-
E\Bigl\|
\hat d_n^{\,oracle}-\theta_n
\Bigr\|^2
\nonumber\\
&\le&
\lambda_{2,n}(2M_n+\lambda_{2,n})
\sum_{j,k}
P\Bigl(
\hat L_{j,k}^{(n)} \neq L_{j,k}^{\star,(n)},\,
(j,k)\in R_u^{(n)}
\Bigr)
\nonumber\\
&\le&
\lambda_{2,n}(2M_n+\lambda_{2,n})
\sum_{j,k}
P\Bigl(
\hat L_{j,k}^{(n)} \neq L_{j,k}^{\star,(n)}
\,\Big|\,
(j,k)\in R_u^{(n)}
\Bigr)
P\Bigl(
(j,k)\in R_u^{(n)}
\Bigr)
\nonumber\\
&\le&
p_n\,
\lambda_{2,n}(2M_n+\lambda_{2,n})
\sum_{j,k}
P\Bigl(
(j,k)\in R_u^{(n)}
\Bigr).
\label{eq:oracle_consistency_proof}
\end{eqnarray}

By \eqref{eq:an_condition}, the factor multiplying $p_n$ is $O(a_n)$, and by \eqref{eq:classifier_consistency}, $p_n\to 0$. Dividing both sides of \eqref{eq:oracle_consistency_proof} by $a_n$ proves \eqref{eq:oracle_consistency}.
\end{proof}

Theorem~\ref{thm:oracle_consistency} should be interpreted modestly. It does not claim global minimax optimality. Rather, it states that if the classifier becomes accurate on the undecided band, then \emph{MLShrink} approaches the best possible keep-or-discard rule within the same two-threshold framework.

\subsection{Bayesian interpretation under a contamination prior}

A complementary interpretation of \emph{MLShrink} is obtained by placing a contamination prior on the wavelet coefficients and studying the Bayes keep-or-discard rule on the undecided band. This viewpoint provides a formal justification for contextual features and clarifies why learning on the intermediate band can improve upon deterministic thresholding.

For each coefficient $(j,k)$, let $H_{j,k}\in\{0,1\}$ denote the latent indicator of signal presence, where $H_{j,k}=0$ corresponds to a noise-only coefficient and $H_{j,k}=1$ corresponds to a signal-bearing coefficient. Conditional on a contextual feature vector $u_{j,k}$, assume
\begin{eqnarray}
P(H_{j,k}=1\mid u_{j,k}) = \pi(u_{j,k}),
\label{eq:prior_inclusion}
\end{eqnarray}

and
\begin{eqnarray}
\theta_{j,k}\mid H_{j,k}=0, u_{j,k} &=& 0,
\nonumber\\
\theta_{j,k}\mid H_{j,k}=1, u_{j,k} &\sim& G_{j,k}(\cdot\mid u_{j,k}),
\label{eq:contamination_prior}
\end{eqnarray}

where $G_{j,k}$ is a slab distribution, possibly depending on scale and context. The observed coefficient satisfies
\begin{eqnarray}
d_{j,k} = \theta_{j,k} + \sigma z_{j,k},
\qquad
z_{j,k}\sim N(0,1).
\label{eq:obs_contamination}
\end{eqnarray}

Under squared-error loss, consider the two available actions on the undecided band: retain the coefficient, yielding estimator $d_{j,k}$, or discard it, yielding estimator $0$. Thus the action is indexed by $a\in\{0,1\}$ and the corresponding estimator is $a\,d_{j,k}$.

\begin{theorem}[Bayes rule under a contamination prior]
\label{thm:bayes_contamination}
Under the contamination prior \eqref{eq:prior_inclusion}--\eqref{eq:obs_contamination}, the posterior risk of action $a\in\{0,1\}$ given $(d_{j,k},u_{j,k})$ is
\begin{eqnarray}
R(a\mid d_{j,k},u_{j,k})
=
E\Bigl[
\bigl(a\,d_{j,k}-\theta_{j,k}\bigr)^2
\mid d_{j,k},u_{j,k}
\Bigr].
\label{eq:posterior_risk_action}
\end{eqnarray}

Let
\begin{eqnarray}
m(d_{j,k},u_{j,k}) = E(\theta_{j,k}\mid d_{j,k},u_{j,k})
\label{eq:posterior_mean}
\end{eqnarray}

denote the posterior mean. Then the Bayes rule retains the coefficient if and only if
\begin{eqnarray}
2\,d_{j,k}\,m(d_{j,k},u_{j,k}) \ge d_{j,k}^2.
\label{eq:bayes_keep_rule}
\end{eqnarray}

Equivalently, if the slab is symmetric and the posterior mean has the same sign as $d_{j,k}$, then the Bayes rule retains the coefficient if and only if
\begin{eqnarray}
|m(d_{j,k},u_{j,k})| \ge \frac{|d_{j,k}|}{2}.
\label{eq:bayes_keep_rule_abs}
\end{eqnarray}
\end{theorem}

\begin{proof}
For $a=1$,
\begin{eqnarray}
R(1\mid d_{j,k},u_{j,k})
&=&
E\Bigl[
(d_{j,k}-\theta_{j,k})^2
\mid d_{j,k},u_{j,k}
\Bigr]
\nonumber\\
&=&
d_{j,k}^2
-
2d_{j,k}E(\theta_{j,k}\mid d_{j,k},u_{j,k})
+
E(\theta_{j,k}^2\mid d_{j,k},u_{j,k}).
\label{eq:r1}
\end{eqnarray}

For $a=0$,
\begin{eqnarray}
R(0\mid d_{j,k},u_{j,k})
=
E(\theta_{j,k}^2\mid d_{j,k},u_{j,k}).
\label{eq:r0}
\end{eqnarray}

Therefore,
\begin{equation}
    \begin{split}
R(1\mid d_{j,k},u_{j,k})-R(0\mid d_{j,k},u_{j,k})
=
d_{j,k}^2 - 2d_{j,k}m(d_{j,k},u_{j,k}).
\label{eq:risk_difference_bayes}
 \end{split}
\end{equation}

Hence retention is preferable if and only if \eqref{eq:bayes_keep_rule} holds. Under the stated sign condition, \eqref{eq:bayes_keep_rule} is equivalent to \eqref{eq:bayes_keep_rule_abs}.
\end{proof}

Theorem~\ref{thm:bayes_contamination} provides a principled interpretation of \emph{MLShrink}. The optimal decision on the undecided band is not determined by magnitude alone, but by the posterior mean of the latent coefficient under a contamination model. Since
\begin{eqnarray}
m(d_{j,k},u_{j,k})
=
P(H_{j,k}=1\mid d_{j,k},u_{j,k})
\,E(\theta_{j,k}\mid H_{j,k}=1,d_{j,k},u_{j,k}),
\label{eq:posterior_mean_factorization}
\end{eqnarray}

both posterior inclusion probability and contextual information influence the decision. This gives a formal justification for using a classifier, or more generally a learned score, to approximate the keep-or-discard rule on the undecided band. This perspective is also consistent with earlier neighboring-coefficient and empirical-Bayes wavelet methods, both of which show that local or levelwise context can materially improve coefficient selection beyond purely marginal thresholding \cite{CaiSilverman2001,JohnstoneSilverman2005}.

\subsection{Utility of contextual features}

The contamination-prior theorem explains why features can matter. The next result shows why same-scale neighborhood information and parent information are particularly natural in wavelet denoising for signals with localized irregularities.

Let $\theta_{j,k}=\langle f,\psi_{j,k}\rangle$ denote the noiseless wavelet coefficients of $f$, where $\psi_{j,k}$ is a compactly supported orthonormal wavelet with $r$ vanishing moments and regularity $r>\alpha$. Let
\begin{eqnarray}
\nu^\theta_{j,k}
=
\max\Bigl\{
|\theta_{j,k-1}|,\,
|\theta_{j,k}|,\,
|\theta_{j,k+1}|
\Bigr\}
\label{eq:nu_theta}
\end{eqnarray}

denote a same-scale neighborhood summary, and let
\begin{eqnarray}
p^\theta_{j,k}
=
|\theta_{j+1,\lfloor k/2\rfloor}|
\label{eq:parent_theta}
\end{eqnarray}

denote the parent magnitude.

\begin{proposition}[Utility of neighborhood and parent features]
\label{prop:feature_utility}
Assume that $f$ is piecewise $C^\alpha$ on $[0,1]$ except for a finite set $\mathcal{S}$ of singular points. Suppose that at each singular point $s\in\mathcal{S}$ the local smoothness index is $\beta_s<\alpha$. Then there exist constants $C>0$, $c_s>0$, and an integer $J_0$ such that for all $j\ge J_0$ the following hold.

If the support of $\psi_{j,k}$ and the supports of its immediate same-scale neighbors are separated from $\mathcal{S}$, then
\begin{eqnarray}
\nu^\theta_{j,k}+p^\theta_{j,k}
\le
C\,2^{-j(\alpha+1/2)}.
\label{eq:smooth_decay}
\end{eqnarray}

On the other hand, for each singular point $s\in\mathcal{S}$ and each sufficiently fine level $j$, there exists an index $k_s(j)$ such that the support of $\psi_{j,k_s(j)}$ intersects a neighborhood of $s$ and
\begin{eqnarray}
\nu^\theta_{j,k_s(j)}+p^\theta_{j,k_s(j)}
\ge
c_s\,2^{-j(\beta_s+1/2)}.
\label{eq:singular_decay}
\end{eqnarray}

Consequently,
\begin{eqnarray}
\frac{
\nu^\theta_{j,k_s(j)}+p^\theta_{j,k_s(j)}
}{
\sup\limits_{k:\,{\rm supp}(\psi_{j,k})\cap \mathcal{S}=\emptyset}
\bigl(\nu^\theta_{j,k}+p^\theta_{j,k}\bigr)
}
\ge
\frac{c_s}{C}\,2^{j(\alpha-\beta_s)},
\label{eq:contrast_growth}
\end{eqnarray}

and the contrast between coefficients associated with singular structure and coefficients from smooth regions increases exponentially with scale.
\end{proposition}

\begin{proof}
The bound \eqref{eq:smooth_decay} is a standard consequence of wavelet regularity theory: away from singularities, compactly supported wavelets with $r>\alpha$ vanishing moments yield coefficient decay of order $2^{-j(\alpha+1/2)}$. Since the same argument applies to neighboring coefficients and to the parent coefficient at level $j+1$, their combined summary is of the same order.

For coefficients whose support intersects a singular point with local smoothness index $\beta_s<\alpha$, wavelet coefficients decay at the slower rate $2^{-j(\beta_s+1/2)}$. Because the wavelet support is compact and dyadic supports are nested across scales, at least one coefficient in a fixed same-scale neighborhood of the singular location, together with its parent coefficient, must inherit this slower decay. This yields \eqref{eq:singular_decay}. Dividing \eqref{eq:singular_decay} by \eqref{eq:smooth_decay} gives \eqref{eq:contrast_growth}.
\end{proof}

Proposition~\ref{prop:feature_utility} gives a mathematical explanation for why neighborhood and parent features are useful. In smooth regions, both same-scale and inter-scale coefficients decay rapidly. Near singularities, however, coefficients persist across adjacent locations and scales. Thus contextual summaries such as neighbor magnitude, parent magnitude, local energy, and short-range inter-scale slopes are not ad hoc additions, but natural descriptors of wavelet-domain signal structure.

The same viewpoint suggests that the feature set used by \emph{MLShrink} can be expanded in several sensible directions. A first extension concerns \emph{local energy} and \emph{local noise scale}. For example, one may use
\begin{eqnarray}
E^{(m)}_{j,k}
=
\sum_{|r-k|\le m} d_{j,r}^2
\label{eq:local_energy_feature}
\end{eqnarray}

as a local energy summary, together with a robust local noise estimate such as
\begin{eqnarray}
\hat\sigma^{\,loc}_{j,k}
=
1.4826\,
{\rm MAD}\Bigl\{
d_{j,r}: |r-k|\le m
\Bigr\},
\label{eq:local_mad_feature}
\end{eqnarray}

which may help detect heteroscedastic or non-i.i.d.\ noise.

A second extension concerns \emph{autocorrelation} and \emph{interscale dependence}. If the noise is colored rather than white, then short-lag autocorrelation and parent-child dependence may themselves be informative. Natural examples are
\begin{eqnarray}
\rho^{(1)}_{j,k}
=
\frac{
\sum_{|r-k|\le m} d_{j,r}d_{j,r-1}
}{
\sum_{|r-k|\le m} d_{j,r}^2+\varepsilon
},
\label{eq:local_ac_feature}
\end{eqnarray}

and
\begin{eqnarray}
\rho^{\,ps}_{j,k}
=
\frac{
d_{j,k}\,d_{j+1,\lfloor k/2\rfloor}
}{
\sqrt{d_{j,k}^2+\varepsilon}\,
\sqrt{d_{j+1,\lfloor k/2\rfloor}^2+\varepsilon}
},
\label{eq:parent_child_corr_feature}
\end{eqnarray}

which measure, respectively, local same-scale autocorrelation and parent-child coherence. These features are especially relevant when the transformed noise is not approximately white within each subband.

A third extension concerns \emph{Bayesian information}. In the contamination-prior setting, it is natural to augment the feature vector by posterior summaries such as the local posterior inclusion probability
\begin{eqnarray}
\hat\pi^{\,post}_{j,k}
=
P\Bigl(
H_{j,k}=1 \mid d_{\mathcal{N}(j,k)},u_{j,k}
\Bigr)
\label{eq:post_inclusion_feature}
\end{eqnarray}

and the local posterior mean
\begin{eqnarray}
\hat m^{\,post}_{j,k}
=
E\Bigl(
\theta_{j,k}\mid d_{\mathcal{N}(j,k)},u_{j,k}
\Bigr),
\label{eq:post_mean_feature}
\end{eqnarray}

where $\mathcal{N}(j,k)$ denotes a local neighborhood in location and scale. Such quantities summarize, in a model-based way, how likely a coefficient is to be signal-bearing and how large the underlying clean coefficient is expected to be.

A fourth extension concerns \emph{local smoothness and regularity}. Beyond the simple parent-child slope already used in \eqref{eq:parent_theta}, one may use local regularity surrogates based on wavelet leaders or related multiscale maxima. For example, letting
\begin{eqnarray}
\ell_{j,k}
=
\sup_{\lambda'\subset 3\lambda_{j,k},\,j'\ge j}
|d_{\lambda'}|,
\label{eq:wavelet_leader_feature}
\end{eqnarray}

one may define the empirical leader slope
\begin{eqnarray}
\alpha^{\,lead}_{j,k}
=
-\log_2
\frac{\ell_{j+1,\lfloor k/2\rfloor}+\varepsilon}{\ell_{j,k}+\varepsilon},
\label{eq:leader_slope_feature}
\end{eqnarray}

which acts as a local smoothness descriptor. Features of this kind may help distinguish isolated noisy coefficients from coefficients generated by genuine singular or cusp-like structure.

Thus, beyond magnitude, level, neighbors, and parents, one may incorporate local energy, robust local noise scale, same-scale autocorrelation, parent-child coherence, posterior inclusion probability, posterior mean, and wavelet-leader-based smoothness surrogates. In practical implementations, these features can be used either individually or in combination, depending on whether the main source of difficulty is non-smooth signal structure, correlated noise, heteroscedasticity, or uncertainty about coefficient inclusion.

\section{Performance Evaluation}\label{sec_perfom_eval}

This section evaluates the empirical performance of \emph{MLShrink} under varying noise levels, threshold choices, signal classes, and classifier selections. We also compare \emph{MLShrink} with a collection of established wavelet-domain denoising rules in order to assess its competitiveness.

The study uses four standard one-dimensional benchmark signals, namely \texttt{Blocks}, \texttt{Heavisine}, \texttt{Doppler}, and \texttt{Bumps}. These signals represent a range of structural behaviors, including smooth oscillation, local irregularity, and sharp discontinuities; see Figure~\ref{fig:test_signals}. Following common practice in the denoising literature, we use the Haar and Daubechies 6-tap wavelet filters for the \texttt{Blocks} and \texttt{Bumps} signals, and the Symmlet 8-tap filter for the \texttt{Heavisine} and \texttt{Doppler} signals.

\begin{figure}
\centering

\begin{subfigure}{0.48\columnwidth}
    \centering
    \includegraphics[width=\linewidth]{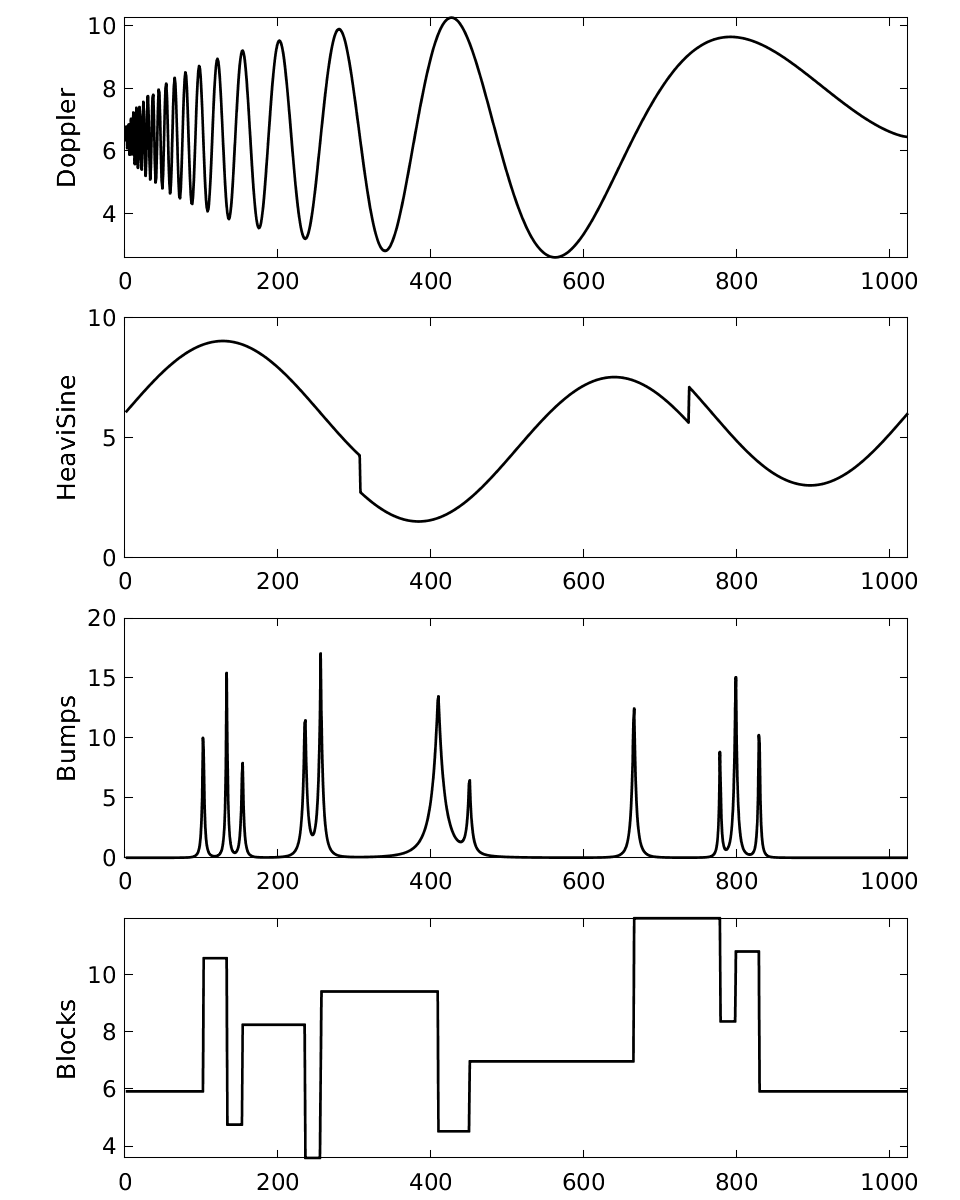}
    \caption{}
    \label{fig:Original_signal}
\end{subfigure}
\hfill
\begin{subfigure}{0.48\columnwidth}
    \centering
    \includegraphics[width=\linewidth]{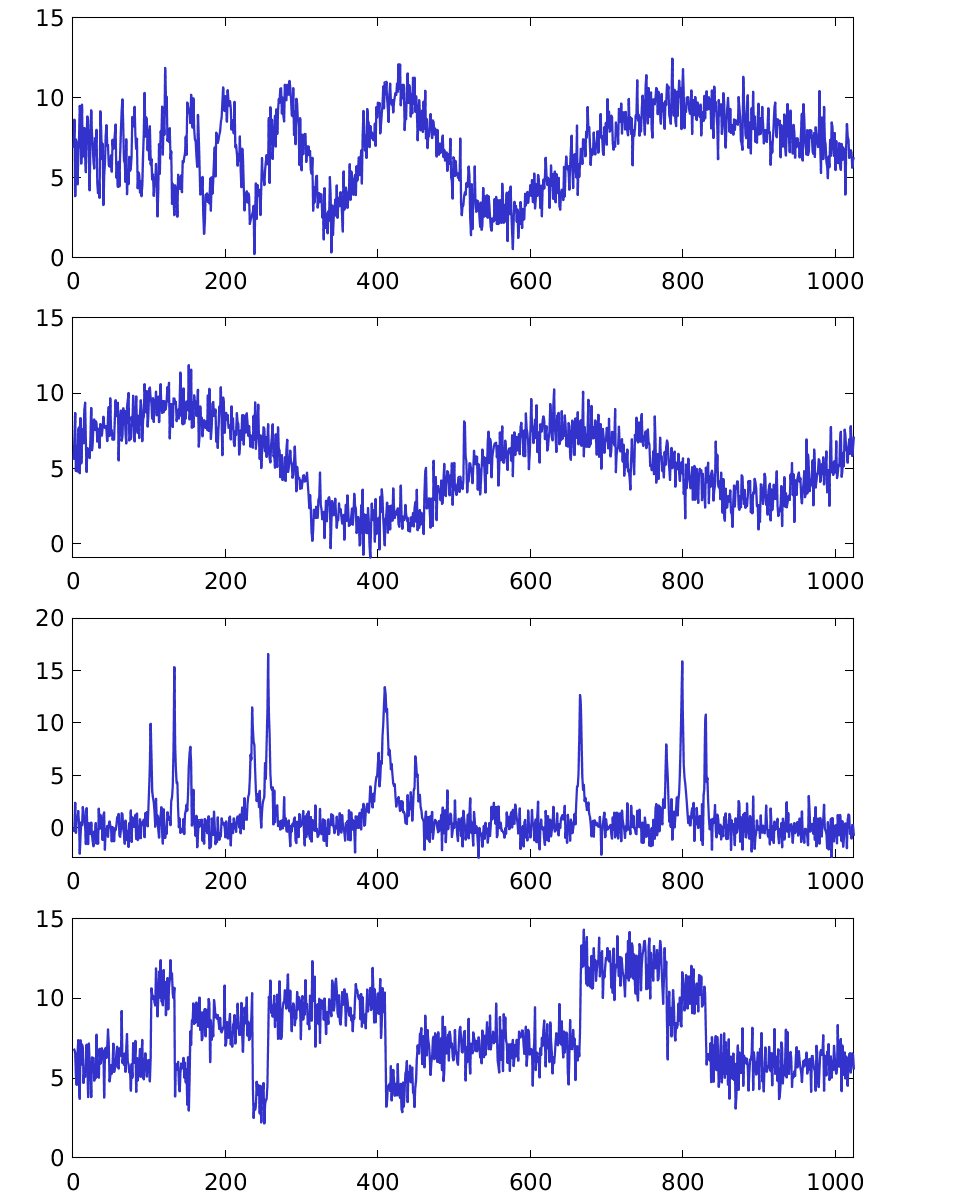}
    \caption{}
    \label{fig:Noisy_signal}
\end{subfigure}
\caption{Four benchmark signals used in the simulation study: (a) clean signals and (b) corresponding noisy realizations.}
\label{fig:test_signals}
\end{figure}

Performance is measured by the average mean squared error (AMSE),
\begin{eqnarray}
AMSE(f) = \frac{1}{nN}\sum_{r=1}^{N}\sum_{i=1}^{n}
\Bigl(f(t_i)-\hat f_r(t_i)\Bigr)^2,
\label{eq:amse}
\end{eqnarray}

where $f$ denotes the underlying signal, $\hat f_r$ is the reconstructed signal in replication $r$, $n$ is the signal length, and $N$ is the number of Monte Carlo replications. In all experiments we use $N=100$ replications.

All simulations were carried out in MATLAB. For the competing denoising rules, we used the GaussianWaveDen toolbox; see also \cite{antoniadis2001wavelet} for related background.
\subsection{Simulation Design}

To generate noisy observations, independent Gaussian noise with variance $\sigma^2=1$ was added to each benchmark signal after rescaling the signal to achieve a prescribed signal-to-noise ratio (SNR). Each noisy realization consisted of $n=1024$ equally spaced samples on $[0,1]$. Figure~\ref{fig:Noisy_signal} shows representative noisy versions of the four test signals.

\subsection{Parameter Selection}\label{sec_hyperParameters}

The empirical performance of \emph{MLShrink} depends primarily on the lower threshold parameter $c$ through
$$
\lambda_1 = \hat \sigma \sqrt{c\log n},
$$

the noise level, the classifier $\mathcal{C}$ used on the undecided band, and the structural features of the underlying signal. The upper threshold is fixed at
$$
\lambda_2 = \hat \sigma \sqrt{2\log n}.
$$

Accordingly, our parameter study focuses on the choice of $\hat \sigma$, the lower-threshold parameter $c$, and the classifier $\mathcal{C}$.

\paragraph{Noise estimation.}
The noise level was estimated from the finest-scale detail coefficients according to
$$
\hat \sigma = \sqrt{\operatorname{Var}(d_{1,k})},
$$

where $d_{1,k}$ denotes the empirical detail coefficients at the finest resolution level and $J=\log_2(n)$.

\paragraph{Selection of the lower threshold.}
To determine a suitable lower threshold, we varied $c$ over the interval $[0.2,2.0]$ and selected the value minimizing AMSE for a fixed signal, classifier, and SNR level. Figure~\ref{fig:Block_performance} illustrates this procedure for the \texttt{Blocks} signal using a Random Forest classifier at SNR $=5$. In that example, the minimum AMSE occurs at $c=1.2$, and the reconstructed signal closely tracks the true signal. The corresponding hyperparameter configuration for the
Random Forest model used in Figure~\ref{fig:Block_performance} is reported in
Table~SI 1 in the Supplementary Information.

\begin{figure}
\centering

\begin{subfigure}{0.475\columnwidth}
    \centering
    \includegraphics[width=1\linewidth]{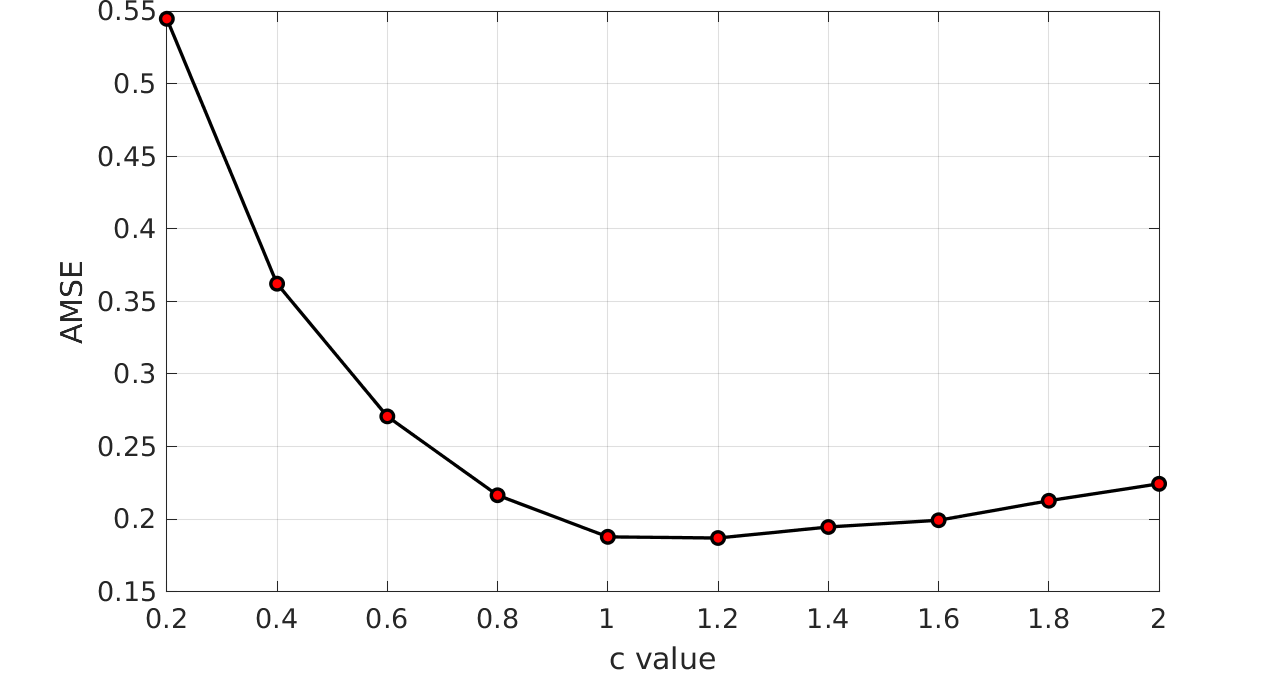}
    \caption{}
    \label{fig:AMSE_with_c}
\end{subfigure}
\hfill
\begin{subfigure}{0.48\columnwidth}
    \centering
    \includegraphics[width=1.1\linewidth]{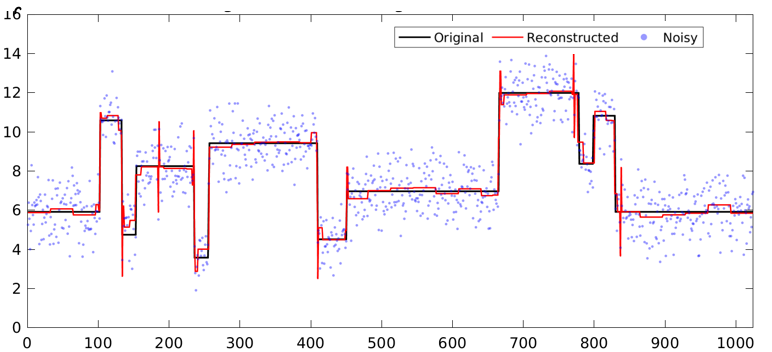}
    \caption{}
    \label{fig:Reconstructed_signal}
\end{subfigure}
\caption{Selection of the lower threshold for \emph{MLShrink}: (a) AMSE as a function of $c$ and (b) reconstructed \texttt{Blocks} signal at the empirically selected value $c=1.2$. The signal, classifier, and SNR were fixed at \texttt{Blocks}, Random Forest, and SNR $=5$, respectively.}
\label{fig:Block_performance}
\end{figure}

\paragraph{Sensitivity to noise level.}
To examine how the choice of $c$ changes with contamination level, we repeated the same analysis for SNR values 3, 5, and 7 while keeping the signal and classifier fixed. Figure~\ref{fig:sensitivity_SNR_change} shows that the AMSE profile as a function of $c$ changes with the noise level, which confirms that the lower threshold should be tuned adaptively rather than fixed a priori across all settings. The optimized parameter settings corresponding to
Figure~\ref{fig:sensitivity_SNR_change} across different SNR levels are reported in
Table~SI 2 in the Supplementary Information.

\begin{figure}
\centering
\includegraphics[width=0.5\textwidth]{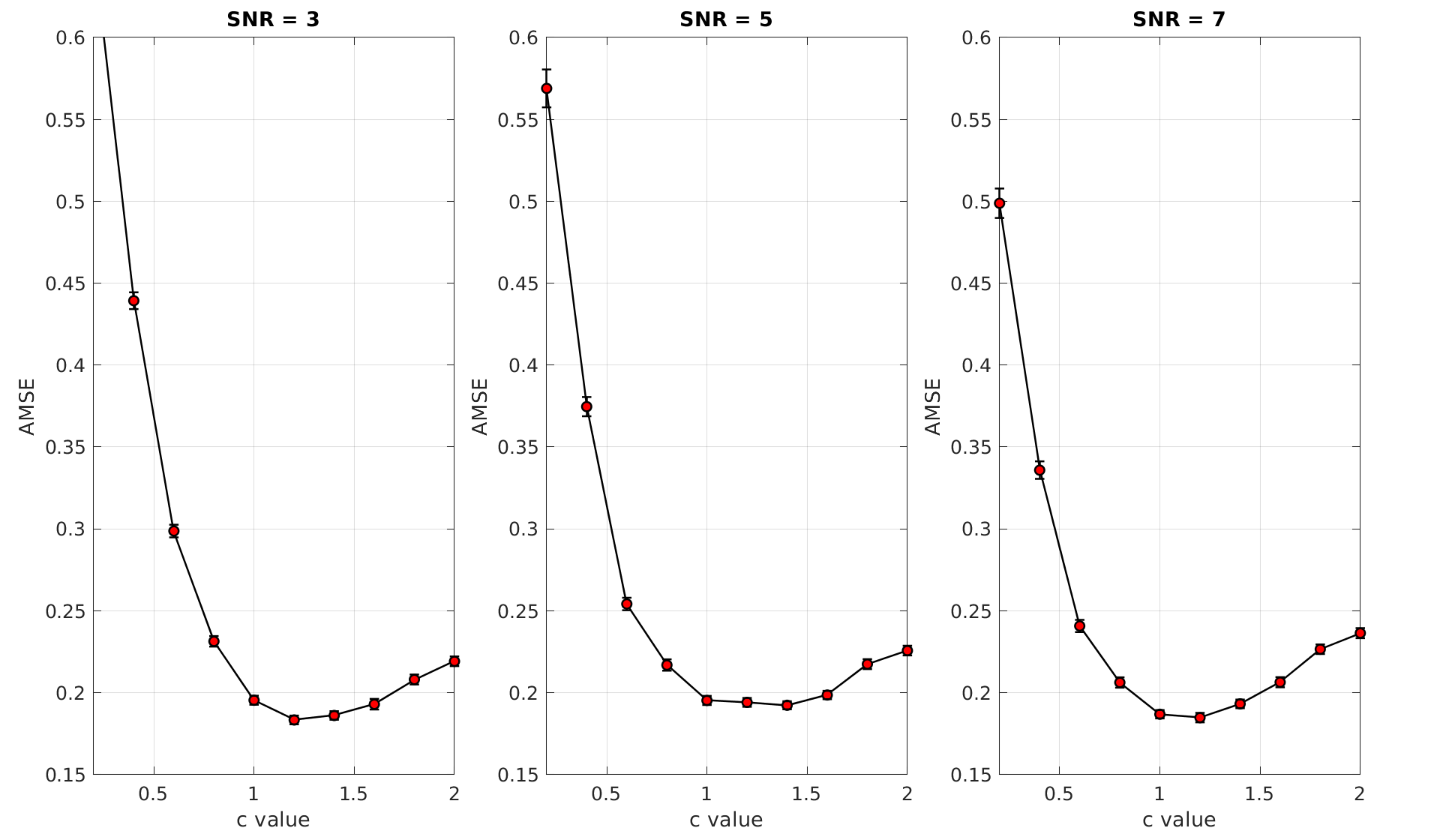}
\includegraphics[width=0.5\textwidth]{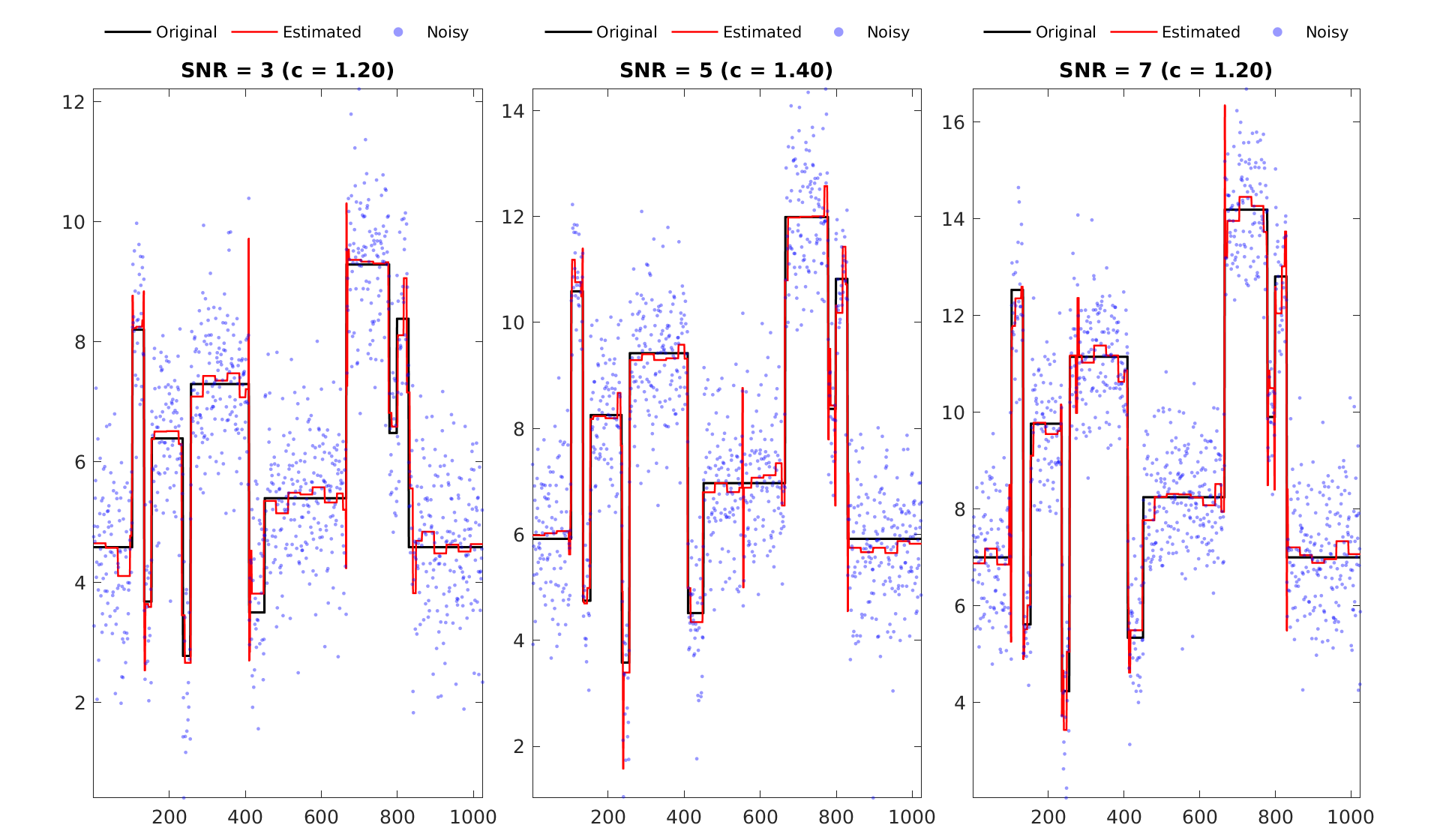}
\caption{Sensitivity of \emph{MLShrink} to noise level for the \texttt{Blocks} signal with a Random Forest classifier. The top panel shows AMSE as a function of $c$, and the bottom panel shows the corresponding reconstructions at the selected values of $c$.}
\label{fig:sensitivity_SNR_change}
\end{figure}

\paragraph{Selection of the classifier.}
The classifier used to label the undecided coefficients is a central component of \emph{MLShrink}. We therefore compared Logistic Regression (\emph{LR}), Support Vector Machines (\emph{SVM}), Random Forests (\emph{RF}), Decision Trees (\emph{DT}), and Neural Networks (\emph{NN}) across a range of values of $c$ and SNR levels. The hyperparameter configurations for each classifier
are summarized in Table~SI 3
in the Supplementary Information. Table~\ref{tab:c-amse} summarizes the best-performing value of $c$ and the corresponding AMSE for the \texttt{Blocks} signal. According to the current numerical results, Random Forest yields the lowest AMSE at SNR $=3$, while Decision Tree yields the lowest AMSE at SNR $=5$ and SNR $=7$. These findings suggest that the optimal classifier may depend on both the signal geometry and the contamination level.

\begin{table*}
\centering
\begin{tabular}{lccccc}
\toprule
 & \multicolumn{5}{c}{\textbf{Classifier}} \\
\cmidrule(lr){2-6}
\textbf{SNR} &
\begin{tabular}[c]{@{}c@{}}\textbf{LR}\\ $c$ / AMSE\end{tabular} &
\begin{tabular}[c]{@{}c@{}}\textbf{SVM}\\ $c$ / AMSE\end{tabular} &
\begin{tabular}[c]{@{}c@{}}\textbf{RF}\\ $c$ / AMSE\end{tabular} &
\begin{tabular}[c]{@{}c@{}}\textbf{DT}\\ $c$ / AMSE\end{tabular} &
\begin{tabular}[c]{@{}c@{}}\textbf{NN}\\ $c$ / AMSE\end{tabular} \\
\midrule
3 & 1.4 / 0.2147 & 1.4 / 0.2164 & {\bf 1.2 / 0.1830} & 1.2 / 0.1883 & 1.4 / 0.1852 \\
5 & 0.2 / 0.2222 & 0.4 / 0.2224 & 1.2 / 0.1851 & {\bf 1.2 / 0.1824} & 1.2 / 0.1828 \\
7 & 0.6 / 0.2306 & 1.0 / 0.2315 & 1.2 / 0.1827 & {\bf 1.2 / 0.1813} & 1.2 / 0.1838 \\
\bottomrule
\end{tabular}
\caption{Best-performing values of $c$ and the corresponding AMSE for the \texttt{Blocks} signal under five classifiers and three SNR levels. Boldface indicates the smallest AMSE within each SNR level.}
\label{tab:c-amse}
\end{table*}

\subsection{Comparative Performance}

We next evaluated \emph{MLShrink} across the four benchmark signals and compared it with a range of classical and modern wavelet shrinkage procedures.

\paragraph{Performance across signal classes.}
For each signal and each SNR level, we selected the empirically best combination of classifier and lower-threshold parameter $c$ from the candidate set considered above. Table~\ref{tab:Comparison} reports the selected classifier, the corresponding $c$ value, and the resulting AMSE for each signal-SNR configuration. Figure~\ref{fig:Reconstruction_SNR5} presents representative reconstructions at SNR = 5.
Additional reconstructions for SNR = 3 and SNR = 7 are
provided in Figure~SI 1 and Figure~SI 2 in the Supplementary Information.

\begin{figure}
	\centering
\includegraphics[width=0.8\textwidth]{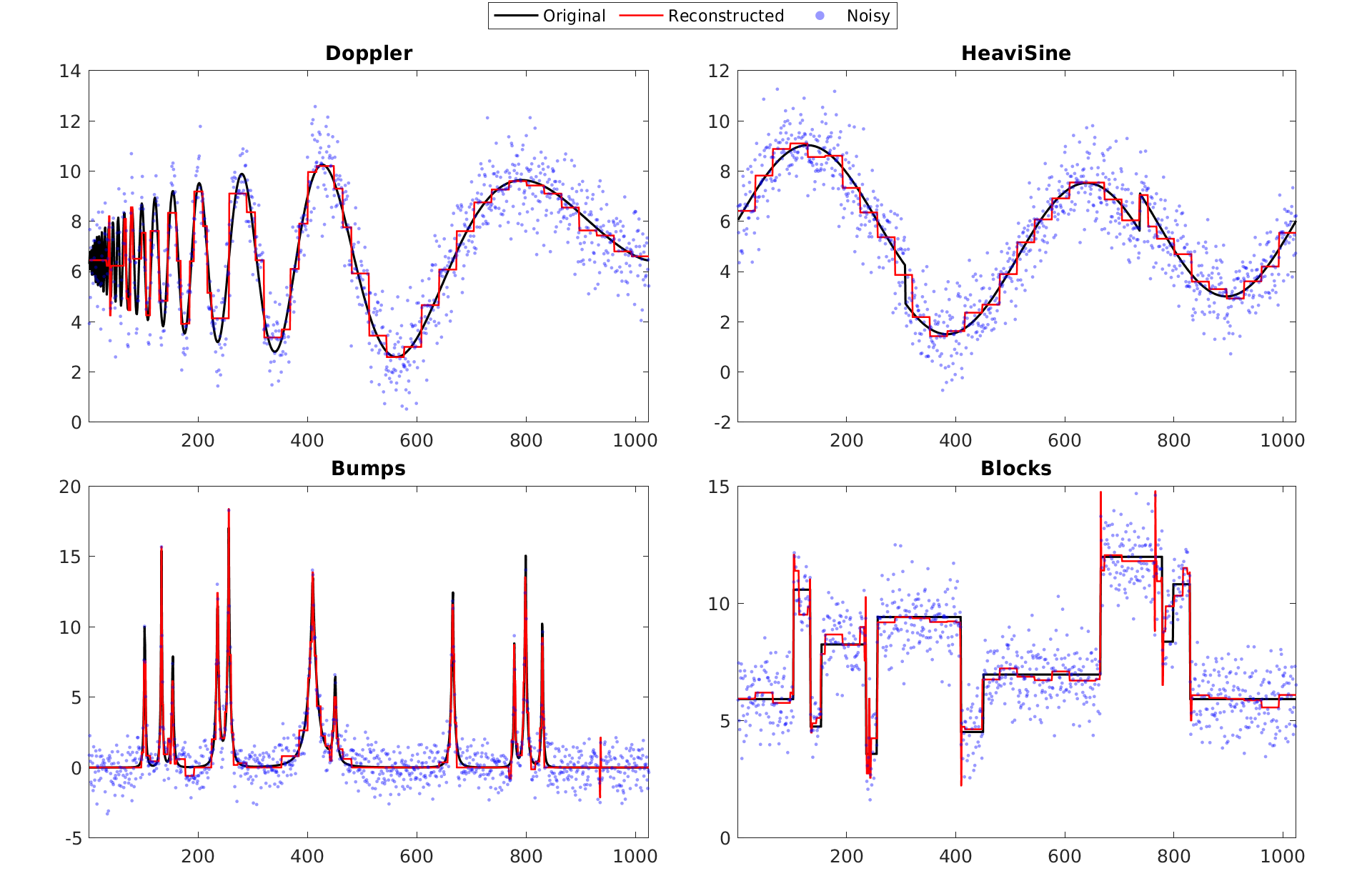}
\caption{Representative \emph{MLShrink} reconstructions for the four benchmark signals at SNR = 5. The black curve is the true signal, the blue points are the noisy observations, and the red curve is the reconstructed signal.}
\label{fig:Reconstruction_SNR5}
\end{figure}

\begin{table*}
\centering
\begin{tabular}{lcccccc}
\toprule
 & \multicolumn{6}{c}{\textbf{SNR}} \\
\cmidrule(lr){2-7}
 & \multicolumn{2}{c}{3} & \multicolumn{2}{c}{5} & \multicolumn{2}{c}{7} \\
\cmidrule(lr){2-3} \cmidrule(lr){4-5} \cmidrule(lr){6-7}
\textbf{Signal} & $c$/Cl. & AMSE & $c$/Cl. & AMSE & $c$/Cl. & AMSE \\
\midrule
Doppler   & 0.80 / LR & 0.1107 & 2.00 / LR & 0.1306 & 1.60 / NN & 0.1495 \\
HeaviSine & 2.00 / DT & 0.0415 & 1.00 / SVM & 0.0464 & 2.00 / NN & 0.0526 \\
Bumps     & 1.00 / NN & 0.3567 & 1.00 / DT & 0.3687 & 1.00 / RF & 0.3845 \\
Blocks    & 1.40 / DT & 0.1799 & 1.20 / DT & 0.1826 & 1.20 / NN & 0.1812 \\
\bottomrule
\end{tabular}
\caption{Empirically selected classifier and lower-threshold parameter for \emph{MLShrink} across the four benchmark signals and three SNR levels.}
\label{tab:Comparison}
\end{table*}

The results indicate that no single classifier dominates uniformly across all settings. This is consistent with the design of \emph{MLShrink}: the learning task on the undecided band depends on the interaction between signal geometry, noise level, and the local feature map. In particular, the results suggest that classifier choice is part of the denoising mechanism rather than a merely secondary tuning component.

\paragraph{Comparison with competing shrinkage rules.} Using the empirically selected values of $c$ and classifier $C$, we compared \emph{MLShrink} with ten competing wavelet-domain procedures: Semi-Soft, Hard, BAMS, DCOMPSH, Block-Median, Block-Mean, Hybrid Block-Median, BlockJS, VisuShrink, and Generalized Cross-Validation. Figure~\ref{fig:Boxplots_SNR5} shows a representative comparison at SNR $=5$. The corresponding comparisons at SNR = 3 and SNR = 7
are provided in Figure~SI 3
and Figure~SI 4 in the Supplementary Information.

\begin{figure}
	\centering
\includegraphics[width=0.8\textwidth]{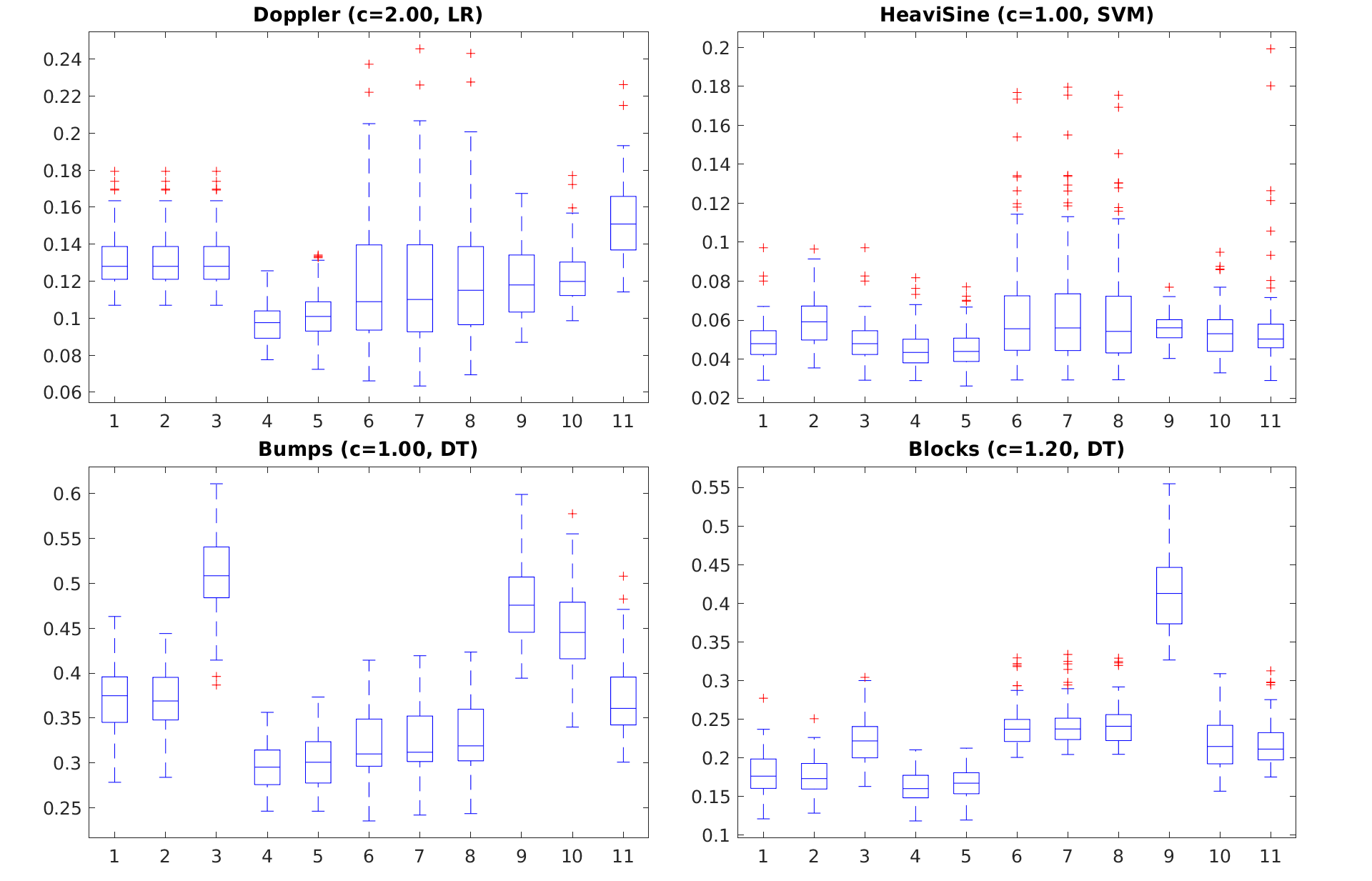}
\caption{Boxplots of AMSE for \emph{MLShrink} and competing denoising procedures at SNR $=5$. The methods shown are: (1) \emph{MLShrink}, (2) Semi-soft, (3) Hard, (4) Bayesian adaptive multiresolution shrinker (BAMS), (5) DCOMPSH, (6) Block-median, (7) Block-mean, (8) Hybrid block-median, (9) BlockJS, (10) VisuShrink, and (11) Generalized cross-validation.}
\label{fig:Boxplots_SNR5}
\end{figure}

Overall, the empirical evidence suggests that \emph{MLShrink} is competitive with the existing shrinkage methods and is particularly promising for signals with pronounced local irregularity or discontinuous structure. For the \texttt{Blocks} signal, \emph{MLShrink} tends to produce smaller reconstruction error than many of the competing methods. For smoother signals such as \texttt{Heavisine}, the differences among methods are narrower, which is consistent with the fact that classical thresholding rules already perform well in smoother settings. More broadly, the simulation results support the interpretation developed in the theoretical section: the principal contribution of \emph{MLShrink} comes from improved decision-making on coefficients in the ambiguous middle band.

The simulation study therefore suggests two main conclusions. First, the performance of \emph{MLShrink} depends materially on both the threshold geometry and the classifier used to resolve the undecided coefficients. Second, despite this dependence, the method remains competitive across a variety of signals and noise levels, with its most visible gains occurring for non-smooth signals where local contextual information is especially informative.

\section{Discussion}\label{sec_discussion}

This paper studied \emph{MLShrink}, a two-threshold wavelet denoising procedure in which coefficients that are clearly small or clearly large are handled deterministically, while coefficients in the intermediate band are classified by a learning rule. The empirical study used four standard benchmark signals, namely \textit{Blocks}, \textit{Bumps}, \textit{HeaviSine}, and \textit{Doppler}, chosen to represent a range of structural behaviors including discontinuities, local spikes, oscillatory structure, and mixed smoothness. Taken together, the numerical results show that \emph{MLShrink} is a competitive denoising procedure and is particularly attractive in settings where a purely magnitude-based thresholding rule may be too rigid.

A central contribution of the present work is that the empirical study is now supported by a theoretical framework tailored to the actual architecture of the method. Rather than treating \emph{MLShrink} as a generic shrinkage family, the theoretical section identifies it as a two-threshold support-selection rule whose only genuinely new statistical component lies in the undecided band. This viewpoint leads to three useful conclusions. First, \emph{MLShrink} is structurally simple: it is nonexpansive in magnitude, sign-preserving under natural feature choices, and identical to a deterministic two-threshold rule outside the intermediate region. Second, its excess quadratic risk relative to an oracle rule is driven entirely by classification errors on the undecided coefficients. Third, under a suitable classifier-consistency assumption, \emph{MLShrink} approaches the oracle performance associated with the same two-threshold architecture. These results do not establish full minimax optimality, but they do provide a mathematically faithful explanation of where the method can help and where it can fail.

The simulation study is consistent with this interpretation. In problems where the keep-or-discard decision is unambiguous, \emph{MLShrink} behaves much like an ordinary thresholding rule. Its advantage appears when a substantial fraction of the coefficients falls into the gray zone between $\lambda_1$ and $\lambda_2$, where local context can improve the decision. This is especially relevant for signals with edges, jumps, transient peaks, or irregular local structure, where neighboring coefficients often carry meaningful information about whether a middle-band coefficient is signal-like or noise-like.

\subsection{Impact of Noise Level and Threshold Geometry}

The numerical experiments confirm that the behavior of \emph{MLShrink} depends strongly on the noise level and on the geometry of the two thresholds. As the contamination level increases, the relative placement of $\lambda_1$ and $\lambda_2$ changes the size of the undecided band and therefore changes both the classification task and the final reconstruction. In lower-SNR settings, more coefficients are pushed toward the lower region, fewer coefficients are confidently retained, and reconstruction tends to become smoother. In higher-SNR settings, more informative coefficients survive into the middle and upper regions, so the classifier has greater opportunity to improve on a purely deterministic rule.

\begin{figure}
	\centering
\includegraphics[width=0.8\textwidth]{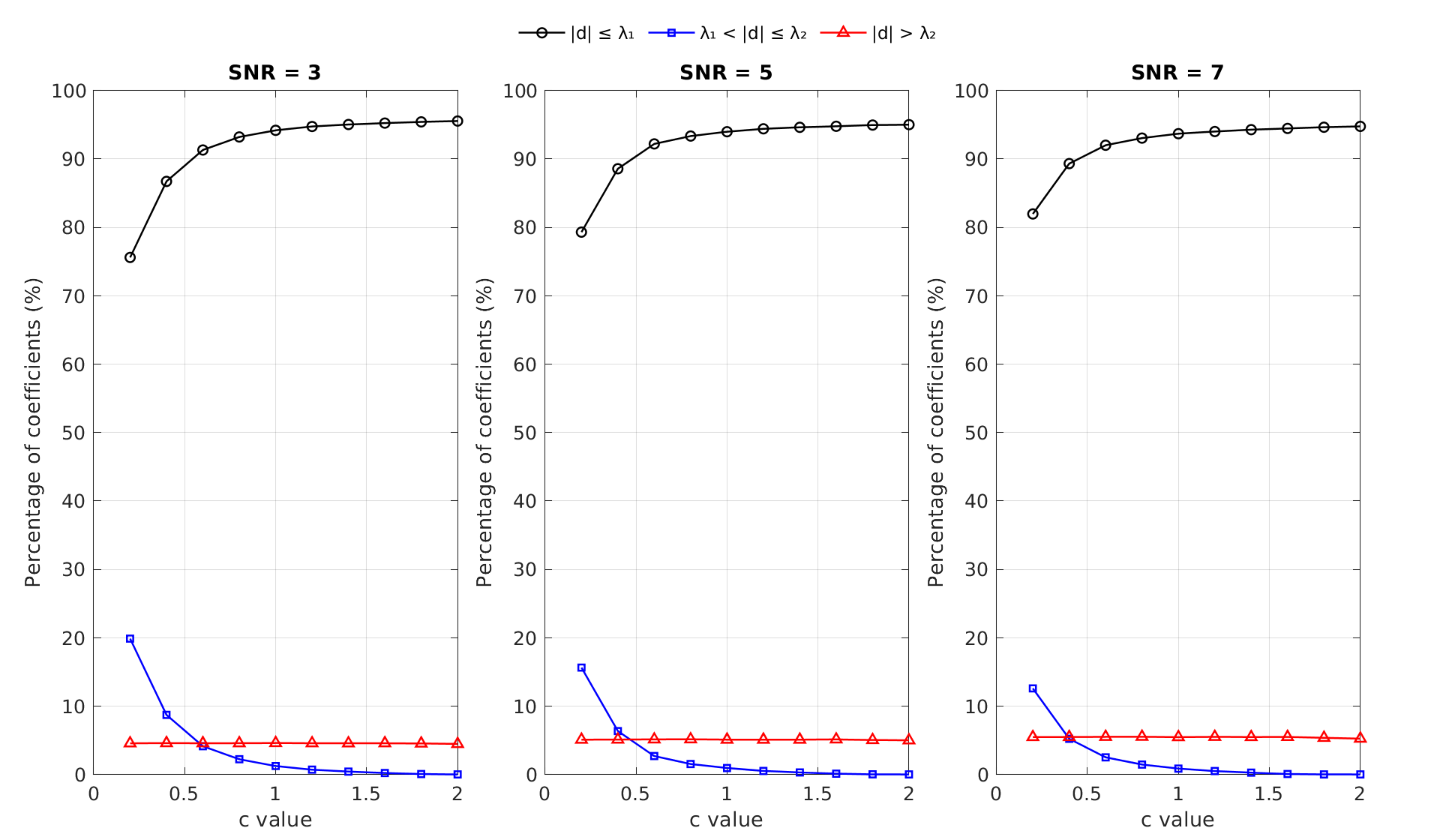}
\caption{Proportion of empirical wavelet coefficients in the three \emph{MLShrink} regions, namely $|d|\leq\lambda_1$, $\lambda_1 < |d| \le \lambda_2$, and $|d|>\lambda_2$, for the \textit{Blocks} signal under SNR $=3,5,7$ and varying values of $c$.}
    \label{fig:c_change_3regions}
\end{figure}

Figure~\ref{fig:c_change_3regions} helps visualize this mechanism. It shows how the proportion of coefficients in the lower, middle, and upper regions changes with $c$ and SNR. This empirical pattern is consistent with the theoretical development in Section~4: the practical behavior of \emph{MLShrink} is governed by how difficult the classification problem becomes on the undecided band and by how informative the available contextual features are in that region. In particular, the size of the undecided band determines how much training information is available and how much scope there is for improvement over deterministic thresholding.

\subsection{Role of Classifier Selection}

The classifier used on the undecided band is not a minor implementation detail but an integral part of \emph{MLShrink}. Different classifiers induce different decision boundaries in the local feature space, and therefore different patterns of coefficient retention. The empirical results suggest that no single classifier dominates uniformly across all signals and SNR levels. This is not surprising. The learning problem faced by \emph{MLShrink} changes with the signal class, the noise level, and the threshold configuration, so the best classifier may also change.

This observation also clarifies the practical meaning of the theory. The theory is intentionally classifier-agnostic: it does not depend on the detailed finite-sample properties of Logistic Regression, Support Vector Machines, Random Forests, Decision Trees, or Neural Networks. Instead, it says that the quality of the final denoiser is controlled by how accurately the classifier recovers the oracle labels on the undecided band. This provides a clean separation between the wavelet-domain architecture of the method and the choice of the learning engine used within that architecture.

\subsection{Comparison with Existing Methods}

The comparison with existing wavelet shrinkage procedures suggests that \emph{MLShrink} is particularly promising for signals with non-smooth structure. For piecewise-constant or edge-dominated signals such as \textit{Blocks}, the method often compares favorably with hard thresholding, semi-soft thresholding, and several other classical procedures. For smoother signals, the differences among methods are smaller, which is also expected: when the signal is sufficiently regular, standard shrinkage rules already perform well and there is less room for a classifier-based correction.

More broadly, the experiments suggest that \emph{MLShrink} should not be viewed as a universal replacement for all denoising rules. Rather, it is a flexible hybrid procedure that appears most useful when coefficient magnitude alone does not provide enough information for a stable decision. In that sense, its best use case is precisely the one emphasized by the theory: settings where the intermediate band contains recoverable signal and local contextual information is informative.

\subsection{Limitations and Future Directions}

The present study still has several limitations. First, the classification step is based on a small feature set consisting mainly of coefficient magnitude, scale information, and a local neighborhood summary. While these features already produce competitive results, they almost certainly do not exhaust the useful information available in the wavelet domain. Richer local and inter-scale descriptors, parent-child relations, wavelet-tree persistence measures, local energy summaries, and covariates derived from the original signal may improve classification accuracy on the undecided band.

Second, the method currently relies on a model-selection stage across several candidate classifiers and threshold parameters. This improves flexibility, but it also increases computational cost. More efficient tuning rules, level-dependent threshold choices, and classifier-specific regularization strategies would make \emph{MLShrink} more scalable.

Third, the present work is restricted to one-dimensional signals and orthonormal wavelet transforms. Extending the methodology to images, volumetric data, wavelet packets, redundant transforms, and nondecimated wavelet systems is a natural next step. Such extensions are particularly appealing because contextual structure across neighboring coefficients is often even stronger in two-dimensional and three-dimensional settings.

A particularly promising future direction is the incorporation of modern deep learning methods into the undecided-band classifier. Instead of using hand-crafted low-dimensional features, one may learn representations directly from local coefficient neighborhoods or from short multiscale coefficient sequences. Convolutional neural networks could be used to learn local wavelet patches, recurrent architectures could capture sequential dependencies across scales, and graph-based neural networks could model parent-child or neighborhood relations in wavelet trees. Recent wavelet-aware deep architectures also suggest concrete paths forward, including multi-level wavelet convolutional networks and learnable wavelet packet models that combine multiscale structure with data-adaptive representation learning \cite{LiuEtAl2018MWCNN,FrusqueFink2024LWPT}.

Transformers are also an appealing possibility. Since wavelet coefficients naturally form structured sequences across locations and scales, transformer-style attention mechanisms may be able to learn long-range dependencies that are invisible to local thresholding rules. A transformer restricted to the undecided band, or supplied with multiscale positional encodings, could in principle learn whether a coefficient should be retained based on a broader contextual field than the current feature map allows. Hybrid procedures are especially attractive here: one may preserve the transparent two-threshold architecture of \emph{MLShrink} while replacing the current shallow classifier by a deeper learned decision module. Such a hybrid would retain interpretability at the thresholding level while gaining representational power in the ambiguous region.

Another important direction is self-supervised or weakly supervised learning in the wavelet domain. Since the initial threshold labels are noisy surrogates rather than true oracle labels, future work could seek iterative refinement schemes, pseudo-labeling strategies, contrastive learning objectives, or uncertainty-aware losses tailored to the undecided band. This may be particularly useful when class imbalance becomes severe or when informative coefficients are rare.

In summary, the present results suggest that \emph{MLShrink} is not only a workable denoising procedure, but also a flexible framework in which wavelet-domain support selection can be coupled with increasingly sophisticated learning mechanisms.

\section{Concluding Remarks}\label{sec_conclusion}

This paper introduced \emph{MLShrink}, a classifier-assisted two-threshold wavelet shrinkage procedure for signal denoising. The method combines the structural simplicity of deterministic thresholding with a data-adaptive learning rule applied only to the coefficients in the undecided band. Coefficients below the lower threshold are discarded, coefficients above the upper threshold are retained, and coefficients in the intermediate region are classified using local wavelet-domain information.

The contribution of the paper is twofold. Methodologically, \emph{MLShrink} offers a flexible alternative to purely magnitude-based thresholding rules by learning the keep-or-discard decision on the undecided band. Theoretically, the paper provides a framework that matches the actual architecture of the method. The structural results show that \emph{MLShrink} remains a nonexpansive, sign-preserving support-selection rule. The oracle risk decomposition shows that its excess risk relative to the best possible middle-band rule is governed entirely by misclassification on the undecided coefficients. The oracle-consistency result further shows that, when the classifier becomes accurate on that band, the denoiser approaches the oracle benchmark associated with the same two-threshold design.

Several directions remain open. These include richer wavelet-domain feature engineering, better handling of class imbalance, more efficient tuning of thresholds and classifiers, extensions to multidimensional data and redundant transforms, and the incorporation of deep learning architectures such as convolutional networks, graph-based models, and transformers into the undecided-band decision step. These developments may substantially broaden the scope of the \emph{MLShrink} framework while preserving the central idea of learned support selection in the wavelet domain.

Overall, \emph{MLShrink} provides a mathematically interpretable and empirically competitive bridge between classical wavelet shrinkage and modern statistical learning. We believe that this hybrid viewpoint offers a useful direction for future research in adaptive denoising. 
\vspace{2mm} \\
\noindent
For reproducibility, the code developed for this study is available at \url{https://github.com/Vijini95/MLShrink}.
\vspace{2mm} \\
\noindent
Supplementary materials are available at:  \url{https://github.com/Vijini95/MLShrink}.
\section*{Acknowledgments}
B. Vidakovic acknowledges the partial support of the H.O. Hartley Chair Foundation and NSF Award 2515246 at Texas A\&M University.

\bibliographystyle{apalike}
\bibliography{reference}

@article{donoho1995adapting,
  title={Adapting to unknown smoothness via wavelet shrinkage},
  author={Donoho, David L and Johnstone, Iain M},
  journal={Journal of the american statistical association},
  volume={90},
  number={432},
  pages={1200--1224},
  year={1995},
  publisher={Taylor \& Francis}
}

@article{donoho1994ideal,
  title={Ideal spatial adaptation by wavelet shrinkage},
  author={Donoho, David L and Johnstone, Iain M},
  journal={biometrika},
  volume={81},
  number={3},
  pages={425--455},
  year={1994},
  publisher={Oxford University Press}
}

@article{donoho1998minimax,
  title={Minimax estimation via wavelet shrinkage},
  author={Donoho, David L and Johnstone, Iain M},
  journal={The annals of Statistics},
  volume={26},
  number={3},
  pages={879--921},
  year={1998},
  publisher={Institute of Mathematical Statistics}
}

@article{antoniadis2001wavelet,
  title={Wavelet estimators in nonparametric regression: a comparative simulation study},
  author={Antoniadis, Anestis and Bigot, Jeremie and Sapatinas, Theofanis},
  journal={Journal of statistical software},
  volume={6},
  pages={1--83},
  year={2001}
}

@article{vimalajeewa2023,
  title={Gamma-minimax wavelet shrinkage for signals with low snr},
  author={Vimalajeewa, Dixon and DasGupta, Anirban and Ruggeri, Fabrizio and Vidakovic, Brani},
  journal={The New England Journal of Statistics in Data Science},
  volume={1},
  number={2},
  pages={159--171},
  year={2023},
  publisher={New England Statistical Society}
}

@incollection{CoifmanDonoho1995TI,
  author    = {Coifman, Ronald R. and Donoho, David L.},
  title     = {Translation-Invariant De-Noising},
  booktitle = {Wavelets and Statistics},
  editor    = {Antoniadis, Anestis and Oppenheim, Georges},
  series    = {Lecture Notes in Statistics},
  volume    = {103},
  pages     = {125--150},
  publisher = {Springer},
  address   = {New York},
  year      = {1995}
}

@article{AbramovichBenjamini1996,
  author  = {Abramovich, Felix and Benjamini, Yoav},
  title   = {Adaptive Thresholding of Wavelet Coefficients},
  journal = {Computational Statistics \& Data Analysis},
  volume  = {22},
  number  = {4},
  pages   = {351--361},
  year    = {1996}
}

@article{GaoBruce1997,
  author  = {Gao, Hong-Ye and Bruce, Andrew G.},
  title   = {WaveShrink with Firm Shrinkage},
  journal = {Statistica Sinica},
  volume  = {7},
  number  = {4},
  pages   = {855--874},
  year    = {1997}
}

@article{AbramovichSapatinasSilverman1998,
  author  = {Abramovich, Felix and Sapatinas, Theofanis and Silverman, Bernard W.},
  title   = {Wavelet Thresholding via a Bayesian Approach},
  journal = {Journal of the Royal Statistical Society: Series B},
  volume  = {60},
  number  = {4},
  pages   = {725--749},
  year    = {1998}
}

@article{HallKerkyacharianPicard1998,
  author  = {Hall, Peter and Kerkyacharian, G{\'e}rard and Picard, Dominique},
  title   = {Block Threshold Rules for Curve Estimation Using Kernel and Wavelet Methods},
  journal = {The Annals of Statistics},
  volume  = {26},
  number  = {3},
  pages   = {922--942},
  year    = {1998}
}

@article{Cai1999,
  author  = {Cai, T. Tony},
  title   = {Adaptive Wavelet Estimation: A Block Thresholding and Oracle Inequality Approach},
  journal = {The Annals of Statistics},
  volume  = {27},
  number  = {3},
  pages   = {898--924},
  year    = {1999}
}

@article{CaiSilverman2001,
  author  = {Cai, T. Tony and Silverman, Bernard W.},
  title   = {Incorporating Information on Neighboring Coefficients into Wavelet Estimation},
  journal = {Sankhy\={a}: The Indian Journal of Statistics, Series B},
  volume  = {63},
  number  = {2},
  pages   = {127--148},
  year    = {2001}
}

@inproceedings{LiuEtAl2018MWCNN,
  author    = {Liu, Pengju and Zhang, Hongzhi and Zhang, Kai and Lin, Liang and Zuo, Wangmeng},
  title     = {Multi-Level Wavelet-CNN for Image Restoration},
  booktitle = {Proceedings of the IEEE/CVF Conference on Computer Vision and Pattern Recognition Workshops},
  year      = {2018},
  pages     = {773--782}
}

@article{FrusqueFink2024LWPT,
  author  = {Frusque, Ga{\"e}tan and Fink, Olga},
  title   = {Robust Time Series Denoising with Learnable Wavelet Packet Transform},
  journal = {Advanced Engineering Informatics},
  volume  = {62},
  pages   = {102669},
  year    = {2024}
}

@article{AntoniadisFan2001,
  author  = {Antoniadis, Anestis and Fan, Jianqing},
  title   = {Regularization of Wavelet Approximations},
  journal = {Journal of the American Statistical Association},
  volume  = {96},
  number  = {455},
  pages   = {939--967},
  year    = {2001}
}

@article{Kulkarnietal2026,
  author  = {Kulkarni, Radhika and Pinheiro, Aluisio and Vidakovic, Brani and Atto, Abdourrahmane M.},
  title   = {Smooth {SCAD}: A Raised Cosine Thresholding Rule for Wavelet Denoising},
  journal = {Mathematics},
  volume  = {14},
  number  = {5},
  pages   = {787},
  year    = {2026},
  doi     = {10.3390/math14050787}
}

@article{KudryavtsevShestakov2024,
  author  = {Kudryavtsev, Alexey A. and Shestakov, Oleg V.},
  title   = {Properties of the SURE Estimates When Using Continuous Thresholding Functions for Wavelet Shrinkage},
  journal = {Mathematics},
  volume  = {12},
  number  = {23},
  pages   = {3646},
  year    = {2024},
  doi     = {10.3390/math12233646}
}

@article{JohnstoneSilverman2005,
  author  = {Johnstone, Iain M. and Silverman, Bernard W.},
  title   = {Empirical Bayes Selection of Wavelet Thresholds},
  journal = {The Annals of Statistics},
  volume  = {33},
  number  = {4},
  pages   = {1700--1752},
  year    = {2005},
  doi     = {10.1214/009053605000000345}
}
\end{document}